\documentclass{article}
\usepackage[utf8]{inputenc}
\usepackage[T1]{fontenc}
\usepackage[a4paper, portrait, margin=1in]{geometry}
\usepackage[dvipsnames]{xcolor}
\usepackage{amsmath,amsfonts,amsthm,booktabs,enumitem,float,array}
\usepackage{graphicx}
\usepackage{mathpazo}
\usepackage{amssymb}
\usepackage{caption}
\usepackage[round]{natbib}
\usepackage{url}
\usepackage{hyperref}
\hypersetup{colorlinks=true, linkcolor=blue, citecolor=teal}
\newcommand{\doi}[1]{doi:~\href{https://doi.org/#1}{\nolinkurl{#1}}}
\usepackage[nameinlink, capitalize, noabbrev]{cleveref}
\usepackage[ruled,vlined,linesnumbered]{algorithm2e}
\newtheorem{theorem}{Theorem}[section]

\newtheorem{definition}{Definition}[section]
\newtheorem{proposition}{Proposition}[section]

\newtheorem{assumption}[theorem]{Assumption}

\newtheorem{remark}[theorem]{Remark}

\providecommand{\argmin}{}
\renewcommand{\argmin}{\operatorname{arg\,min}}

\DeclareMathOperator{\Law}{Law}
\DeclareMathOperator{\Cov}{Cov}

\Crefname{problem}{Problem}{Problems}
\crefname{section}{Section}{Sections}
\Crefname{section}{Section}{Sections}
\crefname{subsection}{Section}{Sections}
\Crefname{subsection}{Section}{Sections}
\crefname{appendix}{Appendix}{Appendices}
\Crefname{appendix}{Appendix}{Appendices}
\crefname{algocf}{Algorithm}{Algorithms}
\crefname{assumption}{Assumption}{Assumptions}
\Crefname{assumption}{Assumption}{Assumptions}
\crefname{remark}{Remark}{Remarks}
\Crefname{remark}{Remark}{Remarks}

\newcommand{\cP}{\mathcal{P}}
\newcommand{\cU}{\mathcal{U}}
\newcommand{\cF}{\mathcal{F}}

\newcommand{\bs}{\mathbf{s}}
\newcommand{\bv}{\mathbf{v}}

\newcommand{\Embb}{\mathbb{E}}

\usepackage{authblk}

\title{Global Multi-Maturity SPX--VIX Calibration Beyond Markovian Stitching}
\author[1]{Atithi Acharya\thanks{These authors contributed equally.\newline Correspondence: \texttt{atithi.acharya@jpmchase.com}, \texttt{charlie.che@jpmchase.com}, \texttt{yue.sun@jpmorgan.com}}}
\author[2]{Charlie Che\protect\footnotemark[1]}
\author[1]{Yue Sun\protect\footnotemark[1]}
\author[1]{\authorcr Brandon Augustino}
\author[1]{Shouvanik Chakrabarti}
\author[1]{Shree Hari Sureshbabu}
\affil[1]{Global Technology Applied Research, JPMorganChase, New York, NY 10001, USA}
\affil[2]{Quantitative Trading \& Research, JPMorganChase, New York, NY 10179, USA}
\date{2026-10-15}

\begin{document}
\maketitle

\begin{abstract}
We develop a global martingale optimal transport framework for joint S\&P~500 (SPX)--VIX smile calibration across multiple maturities, without the conditional independence that Markovian stitching imposes between months.
We prove that exact local and global feasibility are equivalent, yet the stitched laws, which carry earlier months only through the current SPX level, can form a strict subset of the global ones, so adjacent smiles cannot identify the dependence across months.
We show that relative entropy against the standard Markov reference selects the stitched law, so this dependence must be supplied deliberately.
We supply it through a finite-memory reference, in which a latent autoregressive state drives the VIX levels while every VIX marginal stays exact, and calibrate the resulting law on a junction tree whose size grows linearly with the number of maturities.
On market data with up to nine SPX maturities, a law with memory and a law without it fit the same quotes but price VIX-conditioned SPX Asian claims differently: on every calendar, memory makes a claim on the maximum VIX cheaper and a claim on the minimum VIX dearer.
Varying the memory strength on one calendar moves these differences steadily from zero, so the vanilla quotes, which fix the law of each date, do not determine prices that depend on the joint law across dates.
\end{abstract}

\section{Introduction}\label{sec:intro}

The joint calibration of option-pricing models to S\&P~500 (SPX) and VIX options is a central problem in quantitative finance: VIX futures and options are widely used to hedge the volatility exposure of option portfolios, so both markets need to be priced in one model~\citep{cuchiero2025joint}.
The problem is hard because the two markets pull in different directions: in continuous-time models, the steep negative skew of short-dated SPX options implies a very large volatility of volatility, which seems inconsistent with the comparatively low implied volatilities of VIX options~\citep{guyon2019joint}.
A natural nonparametric route is to view joint SPX--VIX calibration through the lens of dispersion-constrained \emph{martingale optimal transport} (MOT)~\citep{Beiglbck2013,guyon2020joint}.
This perspective gives a clean solution in the two-maturity setting, SPX at $T_1$ and $T_2=T_1+30$ days, and VIX at $T_1$, but the multi-maturity problem introduces an identification question that is absent from any one block.
The multi-maturity extension of \citet{guyon2019joint} solves independent two-maturity problems and stitches them recursively through their shared SPX marginal.
This is computationally attractive, but it completes the otherwise unidentified cross-period dependence by assuming that the next SPX--VIX block is conditionally independent of the earlier path given the shared SPX level.
That assumption cannot be tested with the blockwise calibration instruments themselves.
Expiration misalignment creates a separate practical difficulty because no traded SPX marginal is available at a VIX-only date.
Indeed, \citet{guyon2019joint} notes that the rigorous treatment, applying the method to the full vector of all maturities at once, is impractical.
A dense formulation of that problem grows exponentially with the number of maturities.

In this paper we calibrate a single joint law of the SPX and the VIX to all maturities at once, which we call global calibration, and compare it with stitching.
The two approaches agree on feasibility: with exact marginals, a global calibration exists if and only if every monthly, or local, calibration exists.
They differ in identification.
Any global law can be turned into a stitched law with the same monthly laws: keep the law of each monthly block and link consecutive blocks only through the SPX level.
The two laws can still price claims that span several months differently, so the stitched laws can form a strict subset of the global laws.
Among the global laws that fit the same quotes, stitching picks one through its conditional-independence assumption; the prices of claims that span several months therefore carry model risk (\cref{sec:experiments}).
Relative entropy against the standard reference completes the missing dependence with the stitched law, so the dependence must be chosen; we put it in the reference measure, where a latent memory state drives the VIX levels, and a junction-tree representation keeps the size of the calibration linear in the number of maturities.
A secondary point is practical: under bid--ask bands, independently solved blocks on finite grids can select different representatives of a shared SPX smile, and one global solve selects a common one.

\subsection{Contributions}\label{subsec:contributions}

The main contribution is a change in how multi-maturity calibration is set up: one law over the whole path instead of one month at a time joined by stitching.
The modeling results make precise what this changes, and the computational ones make it usable on market data.

\paragraph{Modeling.}
\begin{enumerate}[label=(\roman*)]
    \item \textbf{One law over the whole path.}
    We formulate multi-maturity SPX--VIX calibration as one law over the whole path, with the martingale and VIX constraints conditioned on the full history.
    We prove that exact local and global feasibility are equivalent, and that every global law has a stitched counterpart with the same monthly blocks, its Markovization (\cref{thm:feasibility_equivalence}).
    \item \textbf{Non-identification.}
    We prove that the stitched laws can nevertheless form a strict subset of the global ones: the Markovization preserves every monthly calibration payoff but can change cross-period dependence and prices, and the information it removes is a sum of conditional mutual informations (\cref{thm:strict_inclusion,prop:information_loss}).
    \item \textbf{What entropy selects.}
    We show that the standard Markov reference selects the stitched law as the minimum-information exact KL completion (\cref{prop:kl_markovization}).
    The Markov link across months is therefore set by the reference, not by the quotes, and selecting other dependence requires a history-dependent reference, a cross-period target or a global payoff objective.
    \item \textbf{Bid--ask bands.}
    We characterize compatibility under bid--ask bands: individually feasible relaxed blocks, monthly problems that fit the quotes within their bid--ask bands, need not admit a common shared-marginal selection, and independent solves can choose different representatives of the shared SPX smile; one global solve enforces a common selection (\cref{prop:band_compatibility,subsec:seam}).
    \item \textbf{Misaligned dates.}
    We extend the formulation to misaligned dates using unquoted SPX nodes.
    We give an exact reference construction when the 30-day windows of the VIX contracts do not overlap and, when they overlap, a reference that keeps every VIX marginal and the martingale property while the calibration imposes every window's dispersion condition (\cref{sec:misaligned,prop:merged_reference}).
\end{enumerate}

\paragraph{Computation.}
\begin{enumerate}[label=(\roman*),resume]
    \item \textbf{A reference with memory.}
    We introduce a finite-memory reference in which a latent autoregressive state drives the VIX levels while every VIX marginal stays exact.
    With memory order $b=0$ it is the Markov reference, and calibration against it returns a stitched law (\cref{prop:kl_markovization}); with $b=1$ it carries dependence across months in the VIX path, with a strength set by one parameter, and keeps the SPX kernels of the Markov reference (\cref{subsec:finite_memory}).
    Because the memory is finite, the reference is Markov in a state of fixed size and factorizes along the calendar, and so does the calibrated law; where it exists, the Markovization of the calibrated law splits a price difference between the two orders into a part due to the dependence and a part due to the change of the monthly blocks~\eqref{eq:memory_split}.
    \item \textbf{Gauss--Hermite latent nodes.}
    We discretize the latent state on Gauss--Hermite nodes, which keep its memory strength with few nodes; equal-probability nodes weaken it and need about twice as many nodes to give the same VIX correlation and the same maximum- and minimum-claim prices (\cref{subsec:discretization,app:discretization}).
    \item \textbf{Junction-tree calibration.}
    We calibrate the law on a junction tree whose cliques are the steps of the calendar, so, for a fixed grid and a fixed number of windows open at once, storage and work grow linearly with the number of maturities, including calendars with misaligned and overlapping VIX windows.
    The quoted prices are imposed by Bregman projection and the martingale and dispersion conditions by an exact backward sweep on the tree (\cref{subsec:junction,subsec:solver}).
\end{enumerate}

Together, the two parts show the effect on market data.
On SPX and VIX surfaces with six and nine SPX maturities, on calendars whose VIX dates are not SPX maturities, with and without overlapping VIX windows, a law with memory and a law without it fit the same quotes about equally well but price VIX-conditioned SPX Asian claims differently: on every calendar memory makes a claim on the maximum VIX $19$ to $23$ basis points of notional cheaper and a claim on the minimum VIX $19$ to $72$ basis points dearer.
Memory gives the calibrated law a persistent VIX, which the memoryless law lacks: the VIX levels of consecutive VIX dates have a correlation of $0.34$ to $0.35$ with memory and at most $0.02$ without.
On the calendar without overlapping windows, along a family of laws that fit the same quotes, these differences grow steadily with the memory strength, from zero, where the reference is Markov, to $44$ and $112$ basis points (\cref{sec:experiments}).
The vanilla quotes fix the law of each date but not the joint law across dates on which these claims depend, so that joint law, with its memory, has to be chosen and calibrated as a whole.

\subsection{Related work}\label{subsec:related}

Early parametric approaches, double CEV~\citep{gatheral2008consistent}, regime-switching Heston~\citep{papanicolaou2014regime}, and jump diffusions~\citep{cont2013consistent,baldeaux2014consistent}, offer flexibility but generally fail to jointly calibrate at short maturities.
Rough volatility models~\citep{gatheral2018volatility,jacquier2018vix,jacquier2025rough,gatheral2020quadratic} and neural SDEs~\citep{guyon2023neural} improve accuracy at substantial computational cost.
Computationally, the two-maturity MOT problem is often approached with Sinkhorn-type algorithms~\citep{cuturi2013sinkhorn,altschuler2017near,Benamou2015Iterative,lin2022complexity}; conditional martingale and dispersion constraints require additional projection or optimization steps, and the state dimension grows quickly with the number of maturities.
Augmented-Bregman schemes instead penalize these constraints; as the penalty grows, their solutions converge to the exact relative-entropy projection whenever it exists (\cref{prop:penalty_limit}).
The closest exact-calibration benchmark is \citet{bourgey2024fast}, who develop fast discrete- and continuous-time methods for the joint SPX--VIX smile problem.
\citet{guo2022joint} calibrate a diffusion jointly to SPX and VIX options by semimartingale optimal transport, and \citet{guyon2024dispersion} treats the exact joint calibration with dispersion-constrained martingale Schr\"odinger problems.
Alternative routes use signature methods~\citep{cuchiero2025joint} and Gaussian polynomial volatility models~\citep{abi2025joint}.
Within martingale optimal transport, \citet{demarco2015linking} and \citet{guyon2017bounds} link S\&P~500 smiles to VIX derivatives through model-independent bounds, and \citet{nutz2020multiperiod} study martingale transport over many periods, the setting of our global problem without the VIX; \citet[Section~2.1]{pass2025remarks} characterize how two-period martingale couplings glue into multi-period ones, the product gluing being the only Markovian one.
\citet{che2026spx} compute risk sensitivities of a calibrated SPX--VIX coupling by perturbed optimal transport in the entropic setting, where a skew stickiness ratio for the VIX sets how the VIX smile moves when the SPX moves.
In both settings a snapshot of smiles leaves some dynamics open, which must be supplied from outside the quotes: there how the smiles move with the SPX, here how the VIX depends on its past.
Our block-preserving construction is also related in spirit to Markovian projection and mimicking~\citep{gyongy1986mimicking,brunick2013mimicking}; the precise distinction is given in \cref{rem:mimicking}.

\subsection{Organization}\label{subsec:organization}

\Cref{sec:setting} fixes the setting, notation, and the three constraint families.
\Cref{sec:stitching_to_global} develops the modeling contribution: it establishes feasibility equivalence and strict path-law inclusion, states the no-arbitrage framework, shows that the standard reference selects the stitched law, and distinguishes exact stitching from independent relaxed block solves.
\Cref{sec:misaligned} extends the formulation to misaligned dates on a merged timeline.
\Cref{sec:reference} introduces the finite-memory reference.
\Cref{sec:algorithms} describes the discretization, the junction tree and the solver.
\Cref{sec:experiments} reports the market results.
\Cref{sec:conclusion} concludes.

\section{Setting and Notation}\label{sec:setting}

We work in the setting of \citet{guyon2020joint}, extended to $m$ SPX maturities and $m-1$ VIX maturities.
Fix maturities $T_1 < T_2 < \cdots < T_m$ with $T_{i+1} - T_i = \tau = 30/365$ years.
We assume zero interest rates, repos, and dividends for simplicity.
We write $S_i := S_{T_i}$ for the SPX value at $T_i$, $V_i$ for the VIX at $T_i$, and introduce the strictly convex log-contract payoff
\begin{equation}\label{eq:log_contract}
    L(x) := -\frac{2}{\tau}\log x.
\end{equation}
By definition of the VIX, the price at $T_i$ of the forward-starting log-contract paying $-\frac{2}{\tau}\log\frac{S_{i+1}}{S_i}$ at $T_{i+1}$ is $V_i^2$.

For probability measures $\mu$ and $\nu$, write
\[
D_{\mathrm{KL}}(\mu\Vert\nu)
:=
\begin{cases}
\displaystyle\int\log\!\left(\frac{d\mu}{d\nu}\right)d\mu,
& \mu\ll\nu,\\[4pt]
+\infty, & \text{otherwise}.
\end{cases}
\]

The market data consists of risk-neutral marginals $\mu_{S_i}$ for $i=1,\ldots,m$ (extracted from SPX option prices at each maturity, e.g.\ via Breeden--Litzenberger~\citep{breeden1978prices}) and $\mu_{V_i}$ for $i=1,\ldots,m-1$ (extracted from VIX futures and option prices).
We denote by $S_0 > 0$ the initial SPX value.

\begin{assumption}\label[assumption]{ass:marginals}
The given marginals satisfy
\begin{equation}
    \Embb^{\mu_{S_i}}[S_i] = S_0, \quad \Embb^{\mu_{S_i}}[|\log S_i|] < \infty, \quad i = 1, \ldots, m;
\end{equation}
\begin{equation}
    \Embb^{\mu_{V_i}}[V_i^2] < \infty, \quad i = 1, \ldots, m-1.
\end{equation}
Moreover, the consistency condition
\begin{equation}\label{eq:consistency_global}
    \Embb^{\mu_{V_i}}[V_i^2] = \Embb^{\mu_{S_{i+1}}}[L(S_{i+1})] - \Embb^{\mu_{S_i}}[L(S_i)], \quad i = 1, \ldots, m-1
\end{equation}
holds, and $\mu_{S_i}$, $\mu_{S_{i+1}}$ are in convex order for each $i$.
\end{assumption}
Convex order of two marginals is necessary and sufficient for a martingale coupling of them to exist~\citep{strassen1965existence}.

\paragraph{The three constraint families.} Throughout the paper a calibrated model is a law $\mu$ over the state $(S_1, V_1, S_2, V_2, \ldots, V_{m-1}, S_m)$ that satisfies three families of constraints, which we name once and reuse:
\begin{itemize}
    \item \textbf{(C1) Marginal:} each SPX/VIX smile is matched:
    $S_i \sim \mu_{S_i}$ and $V_i \sim \mu_{V_i}$.
    \item \textbf{(C2) Martingale:} the forward-adjusted conditional mean of the next SPX level equals the current one (here $\Embb^\mu[S_{i+1}\mid\cF_i]=S_i$ under zero rates; in general $=\tfrac{F_{i+1}}{F_i}S_i$).
    \item \textbf{(C3) Dispersion:} the conditional log-contract equals $\mathrm{VIX}^2$:
    $\Embb^\mu[L(S_{i+1}/S_i)\mid\cF_i]=V_i^2$.
\end{itemize}
On a finite grid each family is a finite set of linear constraints on $\mu$.
The single structural distinction between the two frameworks studied here is how adjacent blocks are joined.
Within each monthly block, stitching imposes (C2)--(C3) conditional on $(S_i,V_i)$, but it draws the block $(V_i,S_{i+1})$ from the earlier path using only the shared SPX state $S_i$.
We call the maturity $T_i$ at which two consecutive blocks meet a seam.
The \emph{global} model instead imposes (C2)--(C3) conditional on the entire history $\cF_i=\sigma(S_1,V_1,\ldots,S_i,V_i)$ and does not impose this seam-wise conditional independence.
\Cref{fig:schematic_stitch_vs_global} shows the two conditioning sets.

\begin{figure}[ht]
\centering
\includegraphics[width=\textwidth]{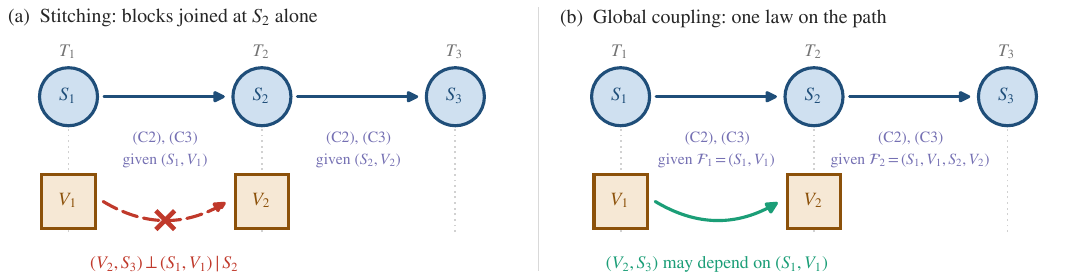}
\caption{What each framework conditions on, for the path $(S_1,V_1,S_2,V_2,S_3)$.
(a) Stitching imposes (C2)--(C3) given $(S_1,V_1)$ and given $(S_2,V_2)$ and joins the second monthly block to the first through $S_2$ alone, so $(V_2,S_3)$ is conditionally independent of $(S_1,V_1)$ given $S_2$.
(b) The global model imposes the same conditions given the full history $\cF_2=\sigma(S_1,V_1,S_2,V_2)$ and no conditional independence.
The two coincide on the first transition and can differ from the second onward.}
\label{fig:schematic_stitch_vs_global}
\end{figure}

\section{From Stitching to Global Coupling}\label{sec:stitching_to_global}

This section motivates the global formulation.
After recalling the stitched construction (\cref{subsec:stitching}), we show that exact stitching imposes a Markov restriction that is invisible to every monthly calibration block (\cref{subsec:stitching_limits}).
We then present the globally coupled model and its duality theory (\cref{subsec:global_model,subsec:duality}), show that relative entropy against the standard reference selects the stitched law (\cref{subsec:markov_reference}), and close with a practical point: independently solved relaxed blocks can select different representatives of a shared SPX smile (\cref{subsec:seam}).

\subsection{The Stitching Approach}\label{subsec:stitching}

For each $i = 1, \ldots, m-1$, the monthly feasible set is defined as in \citet{guyon2020joint}.

\begin{definition}[Monthly feasible set]\label{def:monthly_feasible}
$\cP_i := \cP(\mu_{S_i}, \mu_{V_i}, \mu_{S_{i+1}})$ is the set of all probability measures $\nu_i$ on $\mathbb{R}_{>0} \times \mathbb{R}_{\geq 0} \times \mathbb{R}_{>0}$ such that
\begin{gather}\label{eq:monthly_constraints}
    S_i \sim \mu_{S_i}, \quad V_i \sim \mu_{V_i}, \quad S_{i+1} \sim \mu_{S_{i+1}}, \nonumber\\
    \Embb^{\nu_i}[S_{i+1}|S_i, V_i] = S_i, \qquad \Embb^{\nu_i}\!\left[L\!\left(\frac{S_{i+1}}{S_i}\right)\middle| S_i, V_i\right] = V_i^2.
\end{gather}
\end{definition}

\begin{definition}[Full feasible set]\label{def:full_feasible}
Let $\cP_{\mathrm{full}}$ denote the set of probability measures $\mu$ on the state
\[
(S_1,V_1,S_2,V_2,\ldots,V_{m-1},S_m)
\]
satisfying
\begin{align}
    (C1) &\quad S_i \sim \mu_{S_i},\quad i=1,\ldots,m;
    \qquad V_i \sim \mu_{V_i},\quad i=1,\ldots,m-1, \label{eq:C1}\\
    (C2) &\quad \Embb^\mu[S_{i+1}\mid\cF_i]=S_i,
    \quad i=1,\ldots,m-1, \label{eq:C2}\\
    (C3) &\quad \Embb^\mu\!\left[
    L\!\left(\frac{S_{i+1}}{S_i}\right)\middle|\cF_i
    \right]=V_i^2,
    \quad i=1,\ldots,m-1, \label{eq:C3}
\end{align}
where $\cF_i=\sigma(S_1,V_1,\ldots,S_i,V_i)$.
\end{definition}

\begin{definition}[Stitched model]\label{def:stitched}
Given $\nu_i \in \cP_i$ for each $i = 1, \ldots, m-1$, the stitched model $\mu_{\mathrm{stitch}}$ on $\mathbb{R}_{>0}^{m} \times \mathbb{R}_{\geq 0}^{m-1}$ is defined recursively: $(S_1, V_1, S_2) \sim \nu_1$ and $\Law(V_{i+1}, S_{i+2} \mid S_1, V_1, \ldots, S_{i+1}) = \Law_{\nu_{i+1}}(V_{i+1}, S_{i+2} \mid S_{i+1})$ for $i = 1, \ldots, m-2$.
\end{definition}
This is the multi-maturity construction of \citet{guyon2019joint}.

The critical structural consequence is that the stitched model is \emph{Markovian in the SPX}: the future evolution of $S$ and $V$ depends on the past only through the current SPX value $S_{i+1}$.
All memory of the path, prior VIX levels, prior SPX values, the realized trajectory, is discarded at each stitching boundary.
For each month, \citeauthor{guyon2020joint} builds the calibrating model by solving $\inf_{\nu_i \in \cP_i}D_{\mathrm{KL}}(\nu_i\Vert\bar{\mu}_i)$ with a Sinkhorn-type algorithm; a natural reference uses the lognormal transition kernel
\begin{equation}\label{eq:lognormal_kernel}
    S_{i+1} \mid (S_i = s_i, V_i = v_i) \sim s_i \exp\!\left(v_i\sqrt{\tau}\,G - \tfrac{1}{2}v_i^2\tau\right), \qquad G \sim \mathcal{N}(0,1),
\end{equation}
which automatically satisfies the martingale and dispersion constraints.

\subsection{Structural Limitations of Stitching}\label{subsec:stitching_limits}

The exact distinction is not feasibility but identification of the full path law.
Define
\begin{equation}\label{eq:stitched_class}
\cP_{\mathrm{stitch}}
:=
\left\{
\nu_1(ds_1,dv_1,ds_2)
\prod_{i=2}^{m-1}\nu_i(dv_i,ds_{i+1}\mid s_i)
:\ \nu_i\in\cP_i
\right\}.
\end{equation}
For $\mu\in\cP_{\mathrm{full}}$, let $\nu_i^\mu:=\Law_\mu(S_i,V_i,S_{i+1})$ and define its SPX-Markovization by
\begin{equation}\label{eq:markovization}
\mathsf M\mu
:=
\nu_1^\mu(ds_1,dv_1,ds_2)
\prod_{i=2}^{m-1}\mu(dv_i,ds_{i+1}\mid s_i).
\end{equation}

\begin{theorem}[Feasibility equivalence and block-preserving Markovization]
\label{thm:feasibility_equivalence}
Assume adjacent blocks use the same prescribed shared SPX marginal; regular conditional distributions exist because the state spaces are Euclidean.
For every $\mu\in\cP_{\mathrm{full}}$,
\begin{equation}\label{eq:block_preservation}
    \mathsf M\mu\in\cP_{\mathrm{stitch}}\subseteq\cP_{\mathrm{full}},
    \qquad
    \Law_{\mathsf M\mu}(S_i,V_i,S_{i+1})
    =\Law_\mu(S_i,V_i,S_{i+1})
\end{equation}
for every $i$.
Consequently,
\begin{equation}\label{eq:feasibility_equivalence}
\cP_{\mathrm{full}}\neq\emptyset
\quad\Longleftrightarrow\quad
\cP_{\mathrm{stitch}}\neq\emptyset
\quad\Longleftrightarrow\quad
\cP_i\neq\emptyset\ \text{for every }i.
\end{equation}
\end{theorem}
The proof is given in \cref{app:proof_feasibility_equivalence}.

\begin{remark}[Relation to mimicking]\label[remark]{rem:mimicking}
The map $\mu\mapsto\mathsf M\mu$ is a block-valued discrete-time Markovian projection.
It extends the elementary discrete-time Markovization that preserves adjacent pair laws by preserving the overlapping SPX--VIX blocks $(S_i,V_i,S_{i+1})$.
It is conceptually related to continuous-time mimicking results such as those of \citet{gyongy1986mimicking} and \citet{brunick2013mimicking}, but it is not a diffusion-projection theorem: here the result follows directly from disintegration and Markov gluing.
\end{remark}

\begin{theorem}[Strict law-class inclusion and local non-identification]
\label{thm:strict_inclusion}
For $m\geq3$, there exist marginals satisfying \cref{ass:marginals} for which
\begin{equation}
    \cP_{\mathrm{stitch}}\subsetneq\cP_{\mathrm{full}}.
\end{equation}
More precisely, there is a law $\mu\in\cP_{\mathrm{full}}$ with $\mu\neq\mathsf M\mu$ although all identities in~\eqref{eq:block_preservation} hold.
Thus every integrable block-local payoff $f_i(S_i,V_i,S_{i+1})$ has the same price under $\mu$ and $\mathsf M\mu$, while some bounded cross-period payoff $f$ has different prices under the two laws.
Moreover, $\mu=\mathsf M\mu$ if and only if
\begin{equation}\label{eq:stitched_ci}
    (V_i,S_{i+1})\ \perp\ \cF_{i-1}\mid S_i,
    \qquad i=2,\ldots,m-1,
\end{equation}
up to almost-sure equality of the corresponding kernels.
\end{theorem}
The proof, with the finite-tree example it builds, is given in \cref{app:proof_strict_inclusion}.

\begin{proposition}[Information removed by stitching]\label{prop:information_loss}
Let $\mu\in\cP_{\mathrm{full}}$ and write the pre-$S_i$ history as
\[
H_{i-1}:=(S_1,V_1,\ldots,S_{i-1},V_{i-1}).
\]
Suppose $D_{\mathrm{KL}}(\mu\Vert\mathsf M\mu)<\infty$.
Then
\begin{equation}\label{eq:information_loss}
D_{\mathrm{KL}}(\mu\Vert\mathsf M\mu)
=
\sum_{i=2}^{m-1}
I_\mu\!\left((V_i,S_{i+1});H_{i-1}\mid S_i\right).
\end{equation}
In particular, the information loss is zero exactly when the stitched conditional-independence relations~\eqref{eq:stitched_ci} hold.
\end{proposition}
The proof is given in \cref{app:proof_information_loss}.

The implication for calibration is a negative-identification result.
No collection of payoffs measurable within the individual blocks can test the Markov restriction, because $\mu$ and $\mathsf M\mu$ agree on every such payoff.
Stitching therefore does more than assemble local solutions: it completes the unidentified temporal dependence by imposing~\eqref{eq:stitched_ci}.
This can alter VIX persistence, joint tails, stress exposures, hedge sensitivities, and prices of products spanning several volatility windows.

The proof of \cref{thm:strict_inclusion} builds the finite-tree example: $V_1$ takes two values, $S_2$ takes three levels on each branch, $100$ on both, and $V_2$ is set by $V_1$, so at $S_2=100$ the value of $V_2$ still identifies the branch (\cref{fig:markov_gap}(a)).
Evaluated without numerical optimization, the example and its Markovization satisfy (C2)--(C3) to $1.42\times10^{-14}$.
Markovization attenuates $\Cov(V_1,V_2)$ from $0.040$ to $0.011$, a $72\%$ reduction, and changes the price of the VIX-spread call $(V_2-V_1-K)^+$ from $(0.050,0,0,0)$ to $(0.113,0.063,0.045,0.027)$ at strikes $K=(0,0.1,0.2,0.3)$.
Both laws satisfy the same local constraints and have the same adjacent triple laws; the price difference is entirely due to the cross-period gluing (\cref{fig:markov_gap}).

More generally, for compatible local laws $\nu_1,\ldots,\nu_{m-1}$ define the set of global gluings
\begin{equation}\label{eq:gluing_class}
\Gamma(\nu_1,\ldots,\nu_{m-1})
:=
\left\{\mu\in\cP_{\mathrm{full}}:
\Law_\mu(S_i,V_i,S_{i+1})=\nu_i,
\ i=1,\ldots,m-1\right\}.
\end{equation}
For a cross-period payoff $f$, the interval
\begin{equation}\label{eq:robust_bounds}
\underline p(f):=\inf_{\mu\in\Gamma}\Embb^\mu[f],
\qquad
\overline p(f):=\sup_{\mu\in\Gamma}\Embb^\mu[f]
\end{equation}
measures dependence risk left unidentified by all monthly calibration instruments.
The stitched law supplies one point in this interval, not a model-independent price.
Computing such bounds for market portfolios is left to future work; the finite-tree example establishes that the interval can be non-degenerate.

\begin{figure}[ht]
\centering
\includegraphics[width=\textwidth]{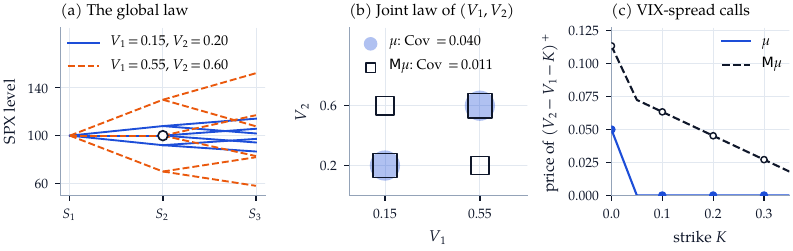}
\caption{Block-preserving Markovization on the exact finite-tree example of \cref{thm:strict_inclusion}.
(a) The global law $\mu$: both branches of $V_1$ reach $S_2=100$ (circled), where $V_2$ still identifies the branch.
(b) The joint law of $(V_1,V_2)$ under $\mu$ and under its Markovization $\mathsf M\mu$, with marker areas proportional to probability: $\Cov(V_1,V_2)$ falls from $0.040$ to $0.011$.
(c) Prices of the VIX-spread calls $(V_2-V_1-K)^+$.
The two laws have identical adjacent calibration blocks; the discrepancy is a pure loss of cross-period information, not a feasibility difference.}
\label{fig:markov_gap}
\end{figure}

\subsection{The Globally Coupled Model}\label{subsec:global_model}

The global model calibrates one law over the entire chain and imposes the martingale and dispersion identities with respect to the full history $\cF_i$ in \cref{def:full_feasible}.
It does not impose the seam-wise conditional independence~\eqref{eq:stitched_ci}; the stitched class is recovered as the subset generated by the Markovization~\eqref{eq:markovization}.
This additional path-law freedom does not change exact local feasibility, but it is necessary for representing cross-period information and for computing dependence-sensitive model-risk bounds.
Let $\Omega$ denote the admissible path space, or its finite grid in a discretization.

\begin{definition}[Full trading universe]\label{def:full_trading}
Let $\cU_{\mathrm{full}}$ denote the set of all portfolios
\[
\Pi = \bigl(\{u_{S_i}\}_{i=1}^{m},\; \{u_{V_i}\}_{i=1}^{m-1},\; \{\Delta_{S,i}\}_{i=1}^{m-1},\; \{\Delta_{L,i}\}_{i=1}^{m-1}\bigr)
\]
where $u_{S_i}$, $u_{V_i}$ are integrable vanilla payoffs and $\Delta_{S,i}$, $\Delta_{L,i}$ are bounded $\cF_i$-measurable dynamic trading strategies.
Its initial cost is
\[
\mathrm{Cost}(\Pi)
:=
\sum_{i=1}^{m}\int u_{S_i}\,d\mu_{S_i}
+\sum_{i=1}^{m-1}\int u_{V_i}\,d\mu_{V_i},
\]
because the dynamic gains have zero initial cost.
The portfolio terminal value is
\begin{multline}\label{eq:portfolio_value}
    \Pi(\bs, \bv) = \sum_{i=1}^{m} u_{S_i}(s_i) + \sum_{i=1}^{m-1} u_{V_i}(v_i) \\
    + \sum_{i=1}^{m-1} \Delta_{S,i}(\bs^i, \bv^i)(s_{i+1} - s_i) + \sum_{i=1}^{m-1} \Delta_{L,i}(\bs^i, \bv^i)\!\left(L\!\left(\frac{s_{i+1}}{s_i}\right) - v_i^2\right).
\end{multline}
Here $\bs^i=(s_1,\ldots,s_i)$ and $\bv^i=(v_1,\ldots,v_i)$.
A \emph{global $(S,V)$-arbitrage} is a portfolio $\Pi \in \cU_{\mathrm{full}}$ with $\Pi(\bs,\bv) \geq 0$ for every $(\bs,\bv)\in\Omega$ and $\mathrm{Cost}(\Pi) < 0$.
\end{definition}

Given a reference measure $\bar{\mu}$, the entropic calibration problem, in the spirit of the relative-entropy calibration of \citet{avellaneda1997calibrating}, is
\begin{equation}\label{eq:entropic_full}
    D_{\bar{\mu}}^{\mathrm{full}} := \inf_{\mu \in \cP_{\mathrm{full}}}D_{\mathrm{KL}}(\mu\Vert\bar{\mu}).
\end{equation}
On a finite grid and under the usual constraint qualification, the optimizer has the normalized exponential form
\begin{multline}
\frac{d\mu^*}{d\bar\mu}(\bs,\bv)
=\frac{1}{Z}\exp\!\Biggl(
\sum_i u^*_{S_i}(s_i)+\sum_i u^*_{V_i}(v_i)\\
+\sum_i\Delta^*_{S,i}(\bs^i,\bv^i)(s_{i+1}-s_i)
+\sum_i\Delta^*_{L,i}(\bs^i,\bv^i)
\left[L(s_{i+1}/s_i)-v_i^2\right]
\Biggr),
\end{multline}
where $Z$ normalizes the law and the dynamic multipliers may depend on the full history.

\subsection{Finite-State Duality and Arbitrage-Freeness}\label{subsec:duality}

The easy direction of the no-arbitrage statement is immediate: any $\mu\in\cP_{\mathrm{full}}$ prices every portfolio in \cref{def:full_trading} at its initial cost and therefore rules out the stated arbitrage.
We state the converse on the finite state spaces used by the numerical checks.
Extending that converse to an unbounded continuum path space requires a separate specification of admissible payoff growth, topology, and closedness; we do not claim that extension here.

\begin{theorem}[Finite-state arbitrage duality]\label{thm:duality}
Every law in $\cP_{\mathrm{full}}$ rules out global $(S,V)$-arbitrage.
On a finite state space the converse holds: if $\cU_{\mathrm{full}}$ admits no global $(S,V)$-arbitrage, then $\cP_{\mathrm{full}}\neq\emptyset$.
\end{theorem}
This is the finite-state form of the martingale-transport duality: the portfolios of \cref{def:full_trading} are exactly the multipliers of the linear constraints (C1)--(C3), so the converse is Farkas' lemma.
We include it for completeness: \cref{thm:entropic_af} uses its first part, and its converse gives the second equivalence in~\eqref{eq:chain}.
The proof is given in \cref{app:proof_duality}.

\begin{theorem}[The entropic model is globally arbitrage-free]\label{thm:entropic_af}
Suppose $\cP_{\mathrm{full}}$ is weakly closed and there exists $\mu_0\in\cP_{\mathrm{full}}$ with $D_{\mathrm{KL}}(\mu_0\Vert\bar\mu)<\infty$.
Then the minimizer
$\mu^*=\argmin_{\mu\in\cP_{\mathrm{full}}}D_{\mathrm{KL}}(\mu\Vert\bar\mu)$ exists and is unique.
As a feasible law, it rules out global arbitrage.
On the finite grids used below, these existence conditions reduce to positivity of the reference on $\Omega$ and nonemptiness of the exact constraint set.
\end{theorem}
The proof is given in \cref{app:proof_entropic_af}.

The chain of implications is
\begin{equation}\label{eq:chain}
    \cP_{\mathrm{stitch}} \neq \emptyset
    \;\;\Longleftrightarrow\;\;
    \cP_{\mathrm{full}} \neq \emptyset
    \;\;\Longleftrightarrow\;\;
    \text{No global arbitrage}
\end{equation}
on a finite state space.
The equivalence concerns nonemptiness, not equality of law classes; \cref{thm:strict_inclusion} shows that the latter can be strict.
\Cref{tab:comparison} summarizes how the two approaches differ.

\begin{table}[ht]
\centering
\caption{Structural comparison of the stitching and globally coupled approaches.}
\label{tab:comparison}
\small
\begin{tabular}{@{}>{\raggedright\arraybackslash}p{4.2cm}>{\raggedright\arraybackslash}p{4.8cm}>{\raggedright\arraybackslash}p{4.8cm}@{}}
\toprule
\textbf{Aspect} & \textbf{Stitching} & \textbf{Globally Coupled} \\
\midrule
Dimensionality per subproblem & 3-dimensional $(S_i, V_i, S_{i+1})$ & $(2m{-}1)$-dimensional full vector \\
\addlinespace
Markov structure & Conditional independence at each SPX seam & No seam-wise Markov assumption \\
\addlinespace
Delta strategies at $T_i$ & $\Delta(s_i, v_i)$ & $\Delta(\cF_i)$ (full history) \\
\addlinespace
Shared SPX maturity & Common in the exact gluing; duplicated in independent relaxed solves & Single shared representation \\
\addlinespace
Joint $(V_i, V_j)$ distribution & Indirect: $V_i \to S_{i+1} \to V_{i+1}$ & Directly modeled \\
\addlinespace
Exact nonemptiness & Nonempty if and only if every monthly problem is & The same condition (\cref{thm:feasibility_equivalence}) \\
\addlinespace
Maturity misalignment & Requires interpolation or endogenous boundary laws & Supports unquoted SPX nodes and overlapping windows (\cref{sec:misaligned}) \\
\bottomrule
\end{tabular}
\end{table}

\subsection{The Markov Reference Selects the Stitched Law}\label{subsec:markov_reference}

The entropic formulation~\eqref{eq:entropic_full} requires a reference measure $\bar{\mu}$ and a feasible law of finite relative entropy with respect to it.
A product-of-marginals reference satisfies the input marginals but generally violates the conditional constraints.
Conversely, a Markov transition reference can satisfy the conditional identities by construction but does not encode non-Markov temporal information.
We build the standard Markov reference and show that it selects the stitched law.

\paragraph{Construction.} The reference interleaves, chronologically, lognormal SPX transition kernels with independent draws of the VIX marginals:
\begin{equation}\label{eq:complete_ref}
    \bar{\mu}(d\mathbf{s}, d\mathbf{v}) = \mu_{S_1}(ds_1) \prod_{i=1}^{m-1}\!\Big[\mu_{V_i}(dv_i) \cdot T_i(s_i, v_i, ds_{i+1})\Big],
\end{equation}
where $T_i$ is the lognormal kernel~\eqref{eq:lognormal_kernel}.
The kernel makes $\Embb^{\bar\mu}[S_{i+1}\mid\cF_i]=S_i$ and $\Embb^{\bar\mu}[L(S_{i+1}/S_i)\mid\cF_i]=V_i^2$ hold exactly, and the independent draws give $V_i\sim\mu_{V_i}$ under $\bar\mu$.

\begin{proposition}[Reference-measure properties]\label{prop:ref_props}
The measure $\bar\mu$ in~\eqref{eq:complete_ref} is well-defined whenever the displayed factors are probability measures, satisfies (C2)--(C3), and matches the $S_1$ and VIX marginals.
Its support is the support induced by $\mu_{S_1}$, the VIX marginals, and the transition kernels.
On a finite grid, positivity of the reference at every point of $\Omega$ implies that every law on that grid is dominated by $\bar\mu$; in the continuum, existence of a feasible finite-entropy law must be assumed separately.
\end{proposition}
The proof is given in \cref{app:proof_entropic_af}.

\begin{proposition}[The standard reference selects the stitched completion]
\label{prop:kl_markovization}
Suppose $\bar\mu$ factorizes through the shared SPX states as in~\eqref{eq:complete_ref}, and let $\mu\in\cP_{\mathrm{full}}$ have finite relative entropy with respect to $\bar\mu$.
Then
\begin{equation}\label{eq:kl_markovization}
D_{\mathrm{KL}}(\mu\Vert\bar\mu)
=
D_{\mathrm{KL}}(\mathsf M\mu\Vert\bar\mu)
+D_{\mathrm{KL}}(\mu\Vert\mathsf M\mu).
\end{equation}
Consequently, every finite-entropy minimizer over $\cP_{\mathrm{full}}$ belongs to $\cP_{\mathrm{stitch}}$.
\end{proposition}
The proof is given in \cref{app:proof_kl_markovization}.

Identity~\eqref{eq:kl_markovization} is the Pythagorean identity of relative entropy~\citep{csiszar1975} in our setting: of all laws with the same monthly blocks as $\mu$, the stitched law $\mathsf M\mu$ is the closest to the reference $\bar\mu$.
Thus the standard reference gives stitching a precise interpretation: it is the minimum-information exact completion of the local calibrations.
This is the block-constrained, discrete-time counterpart of the Markov property of the Schr\"odinger problem, whose solution is Markov when the reference is~\citep[Proposition~2.10]{leonard2014survey}.
The global state space becomes essential when one adds a history-dependent reference, a cross-period target, or a dependence-sensitive payoff objective; it also exposes the range of path laws that the local smiles leave unidentified.

\subsection{Marginal Mismatch at a Seam under Independent Relaxed Calibration}\label{subsec:seam}

Under \cref{def:monthly_feasible}, adjacent \emph{exact} blocks prescribe the same marginal $\mu_{S_{i+1}}$ and therefore admit a consistent gluing.
There is no marginal mismatch at a seam in that idealized definition.
A relaxed block implementation is different: each two-maturity problem is solved on a finite transition grid, with quote moments imposed only on its initial SPX and VIX axes.
The calibrated probabilities and transition rule then imply a terminal SPX distribution, but that distribution is not independently constrained to match the next SPX smile.
A shared maturity $T_{i+1}$ then has two numerical representations, the terminal pushforward of block $i$ and the quote-calibrated initial axis of block $i+1$.
Before those representations are reconciled, the two candidate blocks do not define one exact stitched law.
It is useful to distinguish this mismatch from the structural Markov restriction of \cref{subsec:stitching_limits}.

\begin{itemize}
    \item \emph{Marginal mismatch.} The terminal pushforward of block $i$ may differ from the quote-calibrated initial marginal of block $i+1$.
    \item \emph{Markov restriction.} Even when the shared marginal is exactly common, stitching draws the next block conditionally only on $S_{i+1}$ and therefore discards dependence on the earlier path.
\end{itemize}

\begin{remark}[One shared representation in the global program]
\label[remark]{rem:seam_global}
The globally coupled model of \cref{subsec:global_model} uses a single shared SPX representation at each maturity for both the incoming and the outgoing transition.
It therefore removes the marginal mismatch by construction and does not impose the structural conditional-independence restriction.
\end{remark}

Such independently relaxed blocks need not form one law until their two representations of the shared SPX maturity are reconciled.
A global solve eliminates this duplication by carrying one $S_{i+1}$ variable through both adjacent transitions.
This practical point is separate from exact gluing: when both exact blocks prescribe the same $\mu_{S_{i+1}}$, \cref{thm:feasibility_equivalence} guarantees that they can be glued.

This practical distinction can be stated exactly when quotes are imposed as bid--ask bands.
Let $\mathfrak M_{S_i}$ and $\mathfrak M_{V_i}$ denote the sets of one-dimensional laws whose selected option moments lie inside the corresponding bid--ask intervals, and define the local feasibility relation
\begin{equation}\label{eq:local_feasibility_relation}
\mathfrak R_i
:=
\left\{(\alpha_i,\beta_i,\alpha_{i+1})
\in\mathfrak M_{S_i}\times\mathfrak M_{V_i}\times\mathfrak M_{S_{i+1}}:
\cP(\alpha_i,\beta_i,\alpha_{i+1})\neq\emptyset\right\}.
\end{equation}
Let $\cP_{\mathrm{full}}^{\mathrm{band}}$ be the set of laws satisfying (C2)--(C3) whose one-dimensional marginals belong to these sets.

\begin{proposition}[Compatibility under bid--ask bands]\label{prop:band_compatibility}
The global problem under bid--ask bands is feasible, that is $\cP_{\mathrm{full}}^{\mathrm{band}}\neq\emptyset$, if and only if there is a \emph{single} sequence of marginal laws
\[
(\alpha_1,\beta_1,\ldots,\beta_{m-1},\alpha_m)
\]
such that $(\alpha_i,\beta_i,\alpha_{i+1})\in\mathfrak R_i$ for every $i$.
Equivalently, the compatible fiber product of the relations $\mathfrak R_i$ over their shared SPX marginals is nonempty.
For $m\geq3$, checking only $\mathfrak R_i\neq\emptyset$ separately is not a certificate of global feasibility, because it does not verify that adjacent relations admit the same shared-marginal selection.
\end{proposition}
The proof is given in \cref{app:proof_band_compatibility}.

Thus a jointly optimized global problem can succeed where a \emph{prescribed} sequential rule fails, for example when that rule fixes an interpolated shared marginal outside the compatibility region even though another marginal inside the bid--ask band is compatible with both neighboring blocks.
It cannot succeed where every admissible version of an exact monthly problem is infeasible.
Likewise, introducing the same conditional tolerances on both sides does not reverse the implication: if $Y_i$ denotes either conditional pricing error, then
\[
\left\|\Embb^\mu[Y_i\mid S_i,V_i]\right\|_{L^p}
\leq
\left\|\Embb^\mu[Y_i\mid\cF_i]\right\|_{L^p},
\qquad 1\leq p\leq\infty,
\]
by conditional Jensen.
The stronger practical claim therefore concerns joint selection and minimum relaxation, not an enlargement of exact feasibility.

\section{Extension to Misaligned SPX and VIX Maturities}\label{sec:misaligned}

In practice, the SPX and VIX maturities selected for a calibration need not coincide.
When the calibration instruments do not supply an SPX marginal at a VIX settlement date, a standard fixed-marginal monthly block must either interpolate that missing law or introduce it as an endogenous variable.
The global formulation provides a natural framework on a \emph{merged timeline} $\mathcal{T} = \mathcal{T}_S \cup \mathcal{T}_V = \{t_1 < \cdots < t_N\}$, where $\mathcal{T}_S$ are the SPX expiries and $\mathcal{T}_V=\{T_1^V<\cdots<T_n^V\}$ the settlement dates of the $n$ VIX contracts.
At the VIX-only dates $\mathcal{T}_{\mathrm{free}} := \mathcal{T}_V \setminus \mathcal{T}_S$ no SPX vanilla marginal is quoted, and we let the SPX level there enter as an \emph{unquoted SPX node}: it is free of a marginal constraint, hence the subscript, but it is not unconstrained.
The state is $\mathbf{X} = (S_{t_1}, \ldots, S_{t_N}, V_1, \ldots, V_n)$ with filtration $\cF_k := \sigma(\{S_{t_\ell} : t_\ell \leq t_k\} \cup \{V_i :
T_i^V \leq t_k\})$.

For each $i$, let $T_{\pi(i)}^S$ be the SPX maturity referenced by the $i$-th VIX contract and $\tau_i := T_{\pi(i)}^S - T_i^V$.
For a listed VIX contract, $T_{\pi(i)}^S$ is the SPX expiry $30$ days after $T_i^V$, and we include it in $\mathcal{T}_S$; a window ending at another date would add that date as one more unquoted SPX node.
For each step $[t_k,t_{k+1}]$ of the merged timeline, define
\[
\ell_k := -2\log\!\left(\frac{S_{t_{k+1}}}{d_kS_{t_k}}\right),
\qquad
\mathcal I_i := \left\{k : [t_k,t_{k+1}] \subseteq [T_i^V,T_{\pi(i)}^S]\right\},
\]
where $d_k$ is the forward ratio over $[t_k,t_{k+1}]$.
We call $[T_i^V,T_{\pi(i)}^S]$ the window of the $i$-th VIX contract; $\mathcal I_i$ is the set of its steps.
The dispersion constraint for the $i$-th VIX contract is
\[
\Embb\!\left[
\frac{1}{\tau_i}\sum_{k\in\mathcal I_i}\ell_k
\middle|\cF_{T_i^V}
\right]
=V_i^2.
\]
These constraints are defined whether or not the VIX windows overlap.
There is no marginal constraint at $\mathcal{T}_{\mathrm{free}}$, and martingale constraints are imposed on every step.
\Cref{fig:schematic_misaligned} shows the resulting constraints in the aligned case and in two misaligned cases.
Panels (b) and (c) differ, however, in the construction of a reference measure satisfying all dispersion identities.

\begin{figure}[tbp]
\centering
\includegraphics[width=\textwidth]{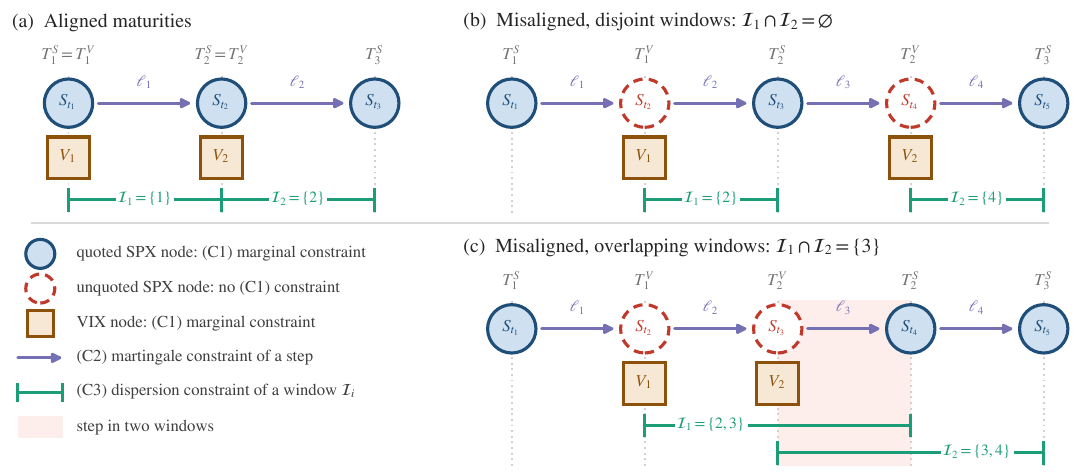}
\caption{Constraints on the merged timeline of SPX expiries $T_j^S$ and VIX dates $T_i^V$; $\mathcal I_i$ is the set of steps in the window of the $i$-th VIX contract.
(a) Aligned maturities: each window is one step.
(b) Misaligned maturities with disjoint windows: the SPX level at a VIX-only date has no (C1) constraint, but it carries the (C2) constraints of both adjacent steps and starts the (C3) constraint of its VIX contract.
Steps $1$ and $3$ lie in no window and carry no (C3) constraint.
(c) Overlapping windows: step $3$ lies in both windows.
The constraints are well defined in every panel; the reference of \cref{prop:merged_reference} gives each step the volatility of the one window that contains it, which (c) does not allow.}
\label{fig:schematic_misaligned}
\end{figure}

\begin{assumption}[Non-overlap condition for the explicit reference]\label[assumption]{ass:non_overlap}
We assume $T_{i+1}^V \geq T_{\pi(i)}^S$ for $i=1,\ldots,n-1$ and $t_1 \in \mathcal{T}_S$.
\end{assumption}

Panels (a) and (b) of \cref{fig:schematic_misaligned} satisfy it; panel (c) does not.
Under \cref{ass:non_overlap}, the following proposition gives an explicit reference construction.

\begin{proposition}[Non-overlapping merged-timeline reference]\label{prop:merged_reference}
Draw $S_{t_1}$ from a prescribed initial law and draw each $V_i$ from its prescribed marginal at time $T_i^V$, independently of the preceding reference history.
On every step in the window of the $i$-th VIX contract, use a conditionally independent lognormal SPX increment with volatility $V_i$ and the step's length; on steps in no window, use any strictly positive martingale transition kernel.
If the VIX windows do not overlap, the chronological product of these kernels is a well-defined martingale reference that matches the initial SPX and all VIX marginals and satisfies
\[
\Embb\!\left[-\frac{2}{\tau_i}\log\!\left(
\frac{S_{T_{\pi(i)}^S}}{d_iS_{T_i^V}}
\right)\middle|\cF_{T_i^V}\right]=V_i^2,
\]
where $d_i$ is the forward ratio over the window.
\end{proposition}
The proof is given in \cref{app:proof_merged_reference}.

The solver of \cref{sec:algorithms} then applies with multi-step dispersion constraints.
The aligned case is recovered when $\mathcal{T}_{\mathrm{free}} = \emptyset$.
The definitions and tower-property arguments extend to this timeline; a continuum no-arbitrage converse would require the additional functional-analytic assumptions noted in \cref{subsec:duality}.

\begin{proposition}[Standard fixed-marginal stitching is not identified at VIX-only dates]\label{prop:stitch_illposed}
If $T_i^V \notin \mathcal{T}_S$, the supplied SPX vanilla marginals do not determine the law $\mu_{S_{T_i^V}}$ required by a standard three-variable block beginning at $T_i^V$.
Consequently, such a block requires either an exogenous interpolation rule or an enlarged formulation that treats the missing SPX law as endogenous.
\end{proposition}
The proof is given in \cref{app:proof_stitch_illposed}.
Such an SPX node is unquoted rather than unconstrained: it still carries a martingale constraint on each adjacent step and enters the dispersion constraint of the VIX contract settling at that date as the base of its log-return, and the entropic objective selects among the laws satisfying those constraints (panel (b) of \cref{fig:schematic_misaligned}).

The full merged-timeline feasible set and its duality mirror \cref{def:full_feasible,thm:duality} with the three changes above; we omit the restatement.

\paragraph{Overlapping windows.}
When two VIX windows overlap, $\mathcal I_i\cap\mathcal I_j\neq\emptyset$, so every $\ell_k$ in the intersection enters both dispersion constraints.
These constraints are conditioned at their respective VIX settlement dates and have right-hand sides $V_i^2$ and $V_j^2$.
Consequently, the prescription used in \cref{prop:merged_reference}, which assigns volatility $V_i$ to every interval in $\mathcal I_i$, is generally overdetermined on the intersection; see panel (c) of \cref{fig:schematic_misaligned}.

The overlapping constraints remain well defined, and the calibration does not need a reference that satisfies them: relative-entropy calibration needs a reference that is positive on the grid and a nonempty constraint set (\cref{thm:entropic_af}), and it imposes every window's dispersion constraint on the calibrated law itself.
We therefore build the reference with a simple allocation of each window's variance: it is split over the window's steps in proportion to their lengths, and a step shared by two windows uses the later window's VIX level.
This keeps every VIX marginal, the martingale property of every step and the latent memory of \cref{sec:reference}, and leaves the earlier window's dispersion identity to the calibration.
The two overlapping calendars of \cref{sec:experiments} are built this way, and their calibrated laws meet every imposed martingale and dispersion constraint to within $2\times10^{-4}$ of its tolerance (\cref{subsec:residuals,tab:fit}).

\section{A Finite-Memory Reference}\label{sec:reference}\label{subsec:finite_memory}

The global formulation can expose a strictly larger class of path laws (\cref{thm:strict_inclusion,prop:information_loss}), carries one shared representation through the full chain (\cref{subsec:seam,rem:seam_global}), and supports unquoted SPX nodes at misaligned dates (\cref{sec:misaligned}).
Two obstacles stand between it and a calibration on market data.
First, relative entropy against the standard reference returns the stitched law (\cref{prop:kl_markovization}), so the larger class is reached only if the reference itself carries memory; this section builds such a reference.
Second, a dense tensor over the whole timeline grows exponentially with the number of maturities.
A reference with finite memory factorizes over the steps of the calendar, and \cref{sec:algorithms} uses that factorization to calibrate on a junction tree whose size grows linearly.

We now put memory into the reference while keeping every VIX marginal exact.
Attach to the $j$-th VIX date a latent memory state $Z_j$ and an emission $W_j$:
\begin{equation}\label{eq:finite_memory}
\begin{aligned}
Z_j &= \varphi Z_{j-1}+\sqrt{1-\varphi^2}\,\varepsilon_j,
&\qquad \varepsilon_j &\overset{\text{iid}}{\sim}\mathcal N(0,1),\quad Z_0\sim\mathcal N(0,1),\\
W_j\mid Z_j &\sim \mathcal N\bigl(\rho_e Z_j,\,1-\rho_e^2\bigr),
&\qquad V_j &= F_{V_j}^{-1}\bigl(\Phi(W_j)\bigr),
\end{aligned}
\end{equation}
with $\varphi\in[0,1)$, $\rho_e\in(0,1)$, $F_{V_j}$ the distribution function of $\mu_{V_j}$ and $\Phi$ the standard normal one.
This is memory order $b=1$; at order $b=0$ the $Z_j$ are independent, and order $b$ would make $(Z_j)$ an AR($b$) process.
The latent state advances once per VIX date, not per unit of time, so $\varphi$ is the correlation of the latent state between consecutive VIX dates, whatever their spacing.
The SPX moves by the kernels of~\eqref{eq:complete_ref} with $V_j$ given by~\eqref{eq:finite_memory}: each SPX step depends only on the current SPX level and, inside a VIX window, on that window's VIX level, never on $Z$ or an earlier VIX level.
We write $\bar\mu_b$ for the resulting reference; \cref{fig:schematic_method}(a) draws its dependence graph, in which no edge leads from $Z$ or an earlier VIX level to the SPX.
This graph is a choice, not something (C2)--(C3) require: a kernel that also read $Z_j$, adjusted to the same constraints, would meet them too, and it would keep the cliques of \cref{subsec:junction}, which already contain $Z_j$.
We make the choice so that $\bar\mu_0$ and $\bar\mu_1$ share every SPX kernel and differ only in the latent transition: the memory is then memory of the VIX path alone, which gives the VIX levels their dependence across dates, and later SPX steps respond to that history only through the current VIX level.

\begin{figure}[t]
\centering
\includegraphics[width=\textwidth]{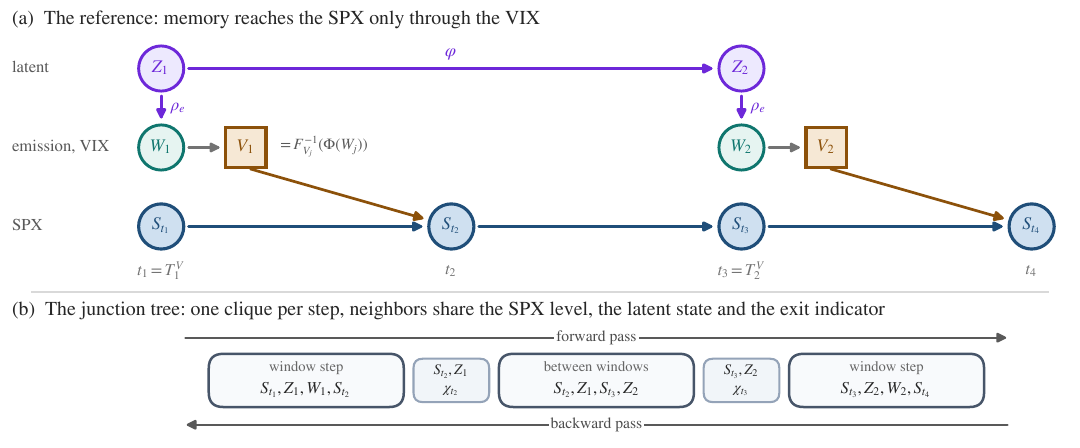}
\caption{The finite-memory reference and its junction tree, for two VIX windows that are single steps and one step between them.
(a) The latent state $Z_j$ moves from one VIX date to the next, the emission $W_j$ depends on $Z_j$, and $V_j=F_{V_j}^{-1}(\Phi(W_j))$; the SPX kernel of a window step depends on the SPX level and $V_j$, never on $Z_j$ or an earlier VIX level.
(b) The discretized law has one clique per step; neighboring cliques share the SPX level, the latent state and the exit indicator $\chi$, which records whether the SPX has left the band (\cref{subsec:discretization}), and one forward and one backward pass give every marginal and conditional moment.}
\label{fig:schematic_method}
\end{figure}

\begin{proposition}[Finite-memory reference]\label{prop:finite_memory}
For every memory order $b\in\{0,1\}$ and every $\varphi$ and $\rho_e$:
\begin{enumerate}[label=(\roman*)]
    \item each $V_j$ has the law $\mu_{V_j}$, and $\bar\mu_b$ satisfies (C2)--(C3);
    \item $\bar\mu_0$ is the Markov reference~\eqref{eq:complete_ref}, so every finite-entropy minimizer against it is stitched;
    \item under $\bar\mu_1$, $\operatorname{corr}(W_i,W_j)=\rho_e^2\varphi^{|i-j|}$ for $i\neq j$.
\end{enumerate}
\end{proposition}
\begin{proof}
Since $Z_j\sim\mathcal N(0,1)$ for every $b$, $W_j$ is a standard normal mixture of $\mathcal N(\rho_e z,1-\rho_e^2)$ over $z\sim\mathcal N(0,1)$ and is standard normal, so $\Phi(W_j)$ is uniform and $V_j\sim\mu_{V_j}$.
The SPX kernels satisfy the martingale and dispersion identities given the current SPX level and VIX level and do not depend on the rest of the history, which gives (C2)--(C3).
At $b=0$ the $V_j$ are independent draws from $\mu_{V_j}$, which is~\eqref{eq:complete_ref}, and \cref{prop:kl_markovization} applies.
At $b=1$, $\operatorname{cov}(W_i,W_j)=\rho_e^2\operatorname{cov}(Z_i,Z_j)=\rho_e^2\varphi^{|i-j|}$ because the emission noises are independent.
\end{proof}
On a merged timeline whose VIX windows do not overlap (\cref{sec:misaligned}), the same construction with any positive martingale kernel on the steps outside the windows keeps (i) and (iii).

Calibrating against $\bar\mu_1$ and against $\bar\mu_0$ with the same constraints gives two laws, $Q_1$ and $Q_0$, that fit the same quotes.
The Markovization $\mathsf MQ_1$ of \cref{thm:feasibility_equivalence}, taken over the steps of the timeline, has the law of $Q_1$ on every step and links consecutive steps only through the state they share: the SPX level and, for the computed laws, the exit indicator of \cref{subsec:discretization}.
For any claim, the price difference therefore splits exactly as
\begin{equation}\label{eq:memory_split}
P_{Q_1}-P_{Q_0}
=\underbrace{\bigl(P_{Q_1}-P_{\mathsf MQ_1}\bigr)}_{\text{memory}}
+\underbrace{\bigl(P_{\mathsf MQ_1}-P_{Q_0}\bigr)}_{\text{block selection}},
\end{equation}
where the first term, the memory term, is the price of the cross-period dependence alone and the second, the block-selection term, is the change in the law of each step that comes with the memory.
The split needs a single SPX node at each seam, so it exists when every VIX window is one step of the timeline.

\section{Calibration on a Junction Tree}\label{sec:algorithms}

The finite-memory reference makes the cross-period dependence explicit, but the calibrated law still lives on the whole timeline.
This section describes how we discretize it (\cref{subsec:discretization}), store it (\cref{subsec:junction}), measure its conditional residuals (\cref{subsec:residuals}) and calibrate it (\cref{subsec:solver}).

\subsection{Discretization}\label{subsec:discretization}

\paragraph{Latent and VIX axes.}
The latent state takes the $n_Z$ Gauss--Hermite nodes and weights of the standard normal, and its transition is the Tauchen--Hussey kernel~\citep{tauchen1991quadrature}
\[
\tilde K(z_a,z_c)\;\propto\;\pi_c\,\frac{g(z_c\mid\varphi z_a,\,1-\varphi^2)}{g(z_c\mid 0,1)},
\]
with node weights $\pi_c$ and $g(\cdot\mid m,s^2)$ the normal density, rescaled by rows and columns until it is row-stochastic and keeps the node weights stationary.
The emission takes $n_W$ equal-probability nodes, the standard normal quantiles at levels $(i-\tfrac12)/n_W$ for $i=1,\ldots,n_W$, each with weight $1/n_W$, and its kernel is built the same way.
The VIX atom of an emission node is the mean of $\mu_{V_j}$ over that node's probability range, so the equally weighted atoms reproduce the VIX forward.
The emission is what lets the two grids differ: the latent state keeps its memory on $n_Z$ Gauss--Hermite nodes, while the VIX takes $n_W$ equally weighted atoms of its own; with $V_j=F_{V_j}^{-1}(\Phi(Z_j))$ the VIX would have only the $n_Z$ latent atoms.
The latent grid sets how much memory the discretized reference carries.
For $\varphi=0.6$, the lag-one correlation of the discretized latent state is $0.639$ at $n_Z=3$, $0.602$ at $n_Z=5$ and $0.600$ from $n_Z=7$ on Gauss--Hermite nodes; equal-probability nodes in the same construction weaken it, to $0.488$, $0.557$ and $0.585$ at $n_Z=3$, $7$ and $20$.
On the non-overlapping calendar of \cref{sec:experiments}, equal-probability nodes need $n_Z=12$ to $15$ to give the VIX correlation and the maximum- and minimum-claim prices that Gauss--Hermite nodes give at $n_Z=7$ (\cref{app:discretization}).

\paragraph{SPX axis.}
At every node of the merged timeline (\cref{sec:misaligned}), the SPX grid has $n_S$ levels evenly spaced in $\log(S/F)$, and the calibration strikes, listed in \cref{subsec:market_setup}, are added at the quoted maturities.
The grid covers the market law, and its half-width grows by at least $1.5$ step standard deviations per step.
Each SPX step is a discretized lognormal step, exponentially tilted until the martingale and dispersion identities hold from every starting level wherever they are attainable: the relative-entropy projection of the step onto its two conditional constraints.

\paragraph{Constraints.}
The market constraints are the forward and the call prices at the calibration strikes of each SPX maturity and of each VIX date, and each window's VIX log-contract identity~\eqref{eq:consistency_global}.
The conditional constraints are the martingale identity of each step and the dispersion identity of each window, one for every conditioning cell.
We impose the conditional constraints only on paths that have never left the band $[0.85\,F_t,\,1.15\,F_t]$ of forward moneyness, so the computed laws satisfy (C2)--(C3) until the SPX first leaves the band, that is, until the first date of the timeline at which it is outside the band, and not after.
To tell these paths apart, the state carries the exit indicator $\chi_{t_k}$: it is $0$ if the SPX has been inside the band at every date of the timeline up to $t_k$, and $1$ otherwise.
No SPX strike outside the band is imposed either; the forward and log-contract constraints involve every path.
We call the SPX levels outside the band the wings.
The constraint weakens along the calendar: under the law with memory on the non-overlapping calendar, the paths that never left the band carry $98.6\%$ of the mass at $24$ days and $46.3\%$ at $269$ days.
A cell of a martingale constraint fixes the SPX level, the latent state and the exit indicator at the start of the step, together with the start state and the VIX node of each window that contains the step; a cell of a window's dispersion constraint fixes the window's start state and VIX node.
Without overlapping windows the calibrated law is Markov in these states (\cref{subsec:junction}), so constraints that hold in every cell hold given the whole path and, by the tower property, given $\cF_i$.
With overlapping windows this still holds for the martingale constraints, but the dispersion constraint of a window that starts while another is open holds given its own start state and VIX node, and not necessarily given the open window's as well; conditioning on both would keep the open window's states in every clique of the later window.
Hence, without overlapping windows, a computed law prices without arbitrage every portfolio whose dynamic positions are held only while $\chi_{t_k}=0$, since each such position has zero expected gain under it.
After the SPX has been outside the band, (C2)--(C3) are not imposed and do not hold (\cref{subsec:market_fit}), so the computed laws are not in $\cP_{\mathrm{full}}$.
\begin{itemize}
    \item \textbf{Why a band:} at every VIX date, the root-mean-square VIX of the input surface exceeds the $30$-day forward volatility of the SPX over its quoted strikes by $1.7$ to $6.2$ volatility points, beyond both bid--ask ranges, with the SPX log-contract at the VIX date interpolated linearly in time between the neighboring expiries (\cref{tab:gap}).
    A law that fits both then takes $15$ to $43\%$ of the VIX variance from SPX levels beyond the quoted strikes, where no quote constrains it.
    \item \textbf{Conditional constraints on every path} are possible but pull against the quotes: the largest errors at the calibration strikes grow by a factor of $3.4$ to $8$, and the price differences of \cref{sec:experiments} keep their signs (\cref{tab:band}).
    \item \textbf{The Markovization of a computed law} keeps the exit indicator with the SPX level: it is~\eqref{eq:markovization} on the steps of the timeline, with $(S_{t_k},\chi_{t_k})$ in place of $S_i$.
    Since $\chi_{t_k}$ is a function of the SPX path, the proof of \cref{thm:feasibility_equivalence} carries over, with steps in place of monthly blocks: this law keeps the joint law of every step and, like the computed law, satisfies (C2)--(C3) until the SPX first leaves the band.
    A Markovization in the SPX level alone would mix paths that have left the band, where no conditional constraint is imposed, with paths that have not.
\end{itemize}
\Cref{fig:support} in \cref{app:discretization} shows the construction on the non-overlapping calendar.

\subsection{The Junction Tree}\label{subsec:junction}

On a calendar whose VIX windows are single steps, the discrete reference is a product of one factor per step: a factor of $(S_{t_k},Z_j,W_j,S_{t_{k+1}})$ on a window step and of $(S_{t_k},Z_j,S_{t_{k+1}},Z_{j+1})$ on a step between windows, with neighboring factors sharing only the SPX level, the latent state and the exit indicator.
These factors are the cliques of a junction tree~\citep{lauritzen1988local} (\cref{fig:schematic_method}(b)): one forward and one backward message pass give every marginal, conditional moment and cell mass exactly, and the joint law is never formed.
Every constraint is a sum of terms that each live on one clique, so the calibrated law has the same cliques; its log-density is the reference's plus one multiplier per constraint times the constraint's terms.
Storage and work per pass therefore grow linearly with the number of steps, and as $n_S^2n_Zn_W$ per window clique (the latent axis is kept at $b=0$ too), instead of exponentially with the number of maturities (\cref{tab:sizes} compares the two).
On calendars whose VIX windows contain other nodes and overlap (\cref{sec:misaligned}), a window spans several steps and neighboring windows share one.
A message then also carries the start state and VIX node of each window open at its step, so its size grows with the number of windows open at once (\cref{tab:sizes}), not with the number of maturities.
With a fixed grid and a fixed number of windows open at once, storage and work therefore still grow linearly with the number of maturities; for monthly VIX contracts, whose $30$-day windows start four to five weeks apart, that number is two.
As with single-step windows, every constraint is a sum of terms that each live on one clique, so the calibrated law keeps these cliques.

\subsection{Conditional Residuals}\label{subsec:residuals}

For the conditional constraint of cell $c$, with centered payoff $g_c$ (the forward-adjusted SPX increment, or the window's log-contract minus $V^2$), let $R_c:=\Embb^Q[g_c\mathbf 1_c]$ be its raw residual under the law $Q$, $\Delta_c$ the cell's volume under the grid's fixed base measure (SPX grid width times node weights), and $D_c$ its economic scale: the forward-adjusted SPX level $d_kS_k$ for a martingale constraint and $V^2$ of the cell's VIX node for a dispersion constraint.
We measure
\begin{equation}\label{eq:dw_residual}
h_c:=\frac{R_c}{\Delta_c D_c},
\qquad
z_c:=\frac{h_c}{\epsilon_c},
\qquad
\mathfrak R_{\mathrm{cond}}:=\max_c\,\lvert z_c\rvert,
\end{equation}
over the constrained cells (those of paths that have never left the band on which the constraint is attainable), with tolerances $\epsilon_{\mathrm{mart}}=0.01$ and $\epsilon_{\mathrm{disp}}=0.05$ frozen before solving; a law meets its conditional constraints when $\mathfrak R_{\mathrm{cond}}\le1$.
Dividing by the base-measure volume, which does not move with the law, keeps the tolerance meaningful across grid resolutions: the raw residual shrinks with the cell mass when the grid is refined, and dividing by the cell mass instead would blow up on the light cells of the wings.

\subsection{Solver}\label{subsec:solver}

Both families are imposed by exact Bregman projections in relative entropy, and the two halves are alternated.
\begin{itemize}
    \item \textbf{Market constraints.} The relative-entropy projection onto one linear constraint is an exponential tilt of the law~\citep{csiszar1975,bauschke1997}; we find its multiplier by a one-dimensional Newton iteration against the current beliefs and visit the constraints cyclically, as in iterative proportional fitting.
    \item \textbf{Conditional constraints.} A cell's martingale and dispersion constraints depend on the law of the path after the cell, not on how the cell was reached.
    A sweep from the last step to the first therefore leaves every cell's constraints exactly satisfied: when a step is projected, the steps after it are already settled, and the projections that follow change only earlier steps.
    A window that spans several steps has its constraints projected together at its first step; across two overlapping windows the sweep is a Gauss--Seidel step, which the residuals of \cref{subsec:residuals} check.
\end{itemize}
A cycle is four passes over the market constraints followed by two backward conditional sweeps; we run $50$ cycles.
The two halves are separate linear families on one law and do not commute, and the alternation converges slowly.
Ending each cycle on the conditional half is a choice.
A law that misses a quote slightly still prices the portfolios of \cref{subsec:discretization} without arbitrage; a law that misses an imposed martingale or dispersion constraint does not.
Every law below therefore meets its conditional constraints, with $\mathfrak R_{\mathrm{cond}}<3\times10^{-4}$ (\cref{tab:fit}), and the remaining error after $50$ cycles lies in the market constraints, mostly in the log-contract constraint of the first window (\cref{subsec:market_fit}).
With $150$ cycles on the non-overlapping calendar that constraint misses by $0.07$ volatility points instead of $1.48$, and the maximum and minimum claims' differences move by less than $0.5$ basis points (\cref{app:discretization}).
When a target cannot be attained on the grid, the exact projection does not exist, and the conditional constraints can be penalized instead; as the penalty grows, the penalized laws converge to the law that misses them least while keeping the market constraints (\cref{prop:penalty_limit}).
None of the laws below needed this.

\section{Market Results}\label{sec:experiments}

We calibrate the finite-memory law of \cref{sec:reference,sec:algorithms} to market data and price claims under it.

\subsection{Data, Calendars and Claims}\label{subsec:market_setup}

\paragraph{Data and grid.}
We use one snapshot of SPX and VIX option chains: the end-of-day quotes of 2026-04-21 from OptionMetrics.
The forward of each SPX maturity and VIX date is implied from its own put and call quotes, so dividends and repo enter through it.
The quotes of each maturity are converted to forward prices with a discount factor implied by put--call parity as well (zero rates of $3.4$ to $4.2\%$; each VIX date takes the SPX rate at its date), and smoothed into an arbitrage-free surface by a P-spline~\citep{eilers1996flexible} in total implied variance against log-moneyness; the constraints and the fit statistics below are computed from these surfaces.
The calibration strikes are $14$ strikes in $[0.85\,F,\,1.15\,F]$ at each SPX maturity, and five strikes from $0.70$ to $1.40$ times the VIX forward at each VIX date (four, from $0.83$, at the first).
The reference has $\varphi=0.6$ and $\rho_e=0.9$, and the grid is $(n_S,n_Z,n_W)=(56,7,12)$, with Gauss--Hermite latent nodes, on every calendar.

\paragraph{Calendars.}
Days are counted from the snapshot.
\begin{itemize}
    \item \textbf{The non-overlapping calendar:} SPX maturities at $24$, $58$, $122$, $178$, $213$ and $269$ days and VIX settlement dates at $28$, $92$, $148$, $183$ and $239$ days, each VIX window ending at the next SPX maturity.
    Every window is one step of the merged timeline of \cref{sec:misaligned} and no two overlap, so \cref{prop:merged_reference} applies and the memory split~\eqref{eq:memory_split} exists.
    \item \textbf{The two overlapping calendars:} SPX maturities at $24$, $58$, $87$, $122$, $150$ and $178$ days and consecutive monthly VIX settlement dates at $28$, $57$, $92$, $120$ and $148$ days; the nine-maturity calendar adds SPX maturities at $213$, $241$ and $269$ days and VIX dates at $183$, $211$ and $239$ days.
    Every window contains another node, and three and five pairs of consecutive windows overlap by one step.
\end{itemize}
\Cref{tab:calendars} lists them, and \cref{app:computation} gives the size of each calibration.
With the reference and grid above, $Q_1$ and $Q_0$ of each calendar form its baseline pair, and $Q_1$ its baseline law; without a named calendar, the baseline pair and law are those of the non-overlapping calendar.

\begin{table}[ht]
\centering
\caption{The three calendars, named by whether their VIX windows overlap and, in parentheses, by the number of SPX maturities: the number and range, in days from 2026-04-21, of the SPX maturities and of the VIX settlement dates, and the market constraints.}
\label{tab:calendars}
\small
\begin{tabular}{@{}lccr@{}}
\toprule
calendar & SPX maturities & VIX dates & market constraints \\
\midrule
non-overlapping & 6 (24--269) & 5 (28--239) & 124 \\
overlapping (6) & 6 (24--178) & 5 (28--148) & 124 \\
overlapping (9) & 9 (24--269) & 8 (28--239) & 190 \\
\bottomrule
\end{tabular}
\end{table}

\paragraph{Claims.}
The claims are four VIX-conditioned SPX Asian payoffs of notional $1$, paid at the last SPX fixing.
Let $A_S$ be the equally weighted average of $S_t/F_0(t)$ over the calendar's SPX fixings $t$, which are the window ends; let $A_V$, $M_V$ and $m_V$ be the average, maximum and minimum of the VIX over its settlement dates; and write each barrier as a multiple of $\bar F_V$, the mean quoted VIX forward over those dates.
Then
\begin{equation}\label{eq:claims}
\begin{aligned}
  H_{\mathrm{smooth}} &= L_S\, G_V, &
  H_{\mathrm{avg}} &= (1 - A_S)^{+}\, \mathbf{1}\{A_V \ge B_{\mathrm{avg}}\},\\
  H_{\mathrm{max}} &= (1 - A_S)^{+}\, \mathbf{1}\{M_V \ge B_{\mathrm{max}}\}, &
  H_{\mathrm{min}} &= (1 - A_S)^{+}\, \mathbf{1}\{m_V \ge B_{\mathrm{min}}\},
\end{aligned}
\end{equation}
where $L_S$ is the put spread on $A_S$ between $0.85$ and $1$, normalized to fall from $1$ to $0$, and $G_V$ is the call spread on $A_V$ between $1.05\,\bar F_V$ and $1.20\,\bar F_V$, normalized to rise from $0$ to $1$.
Each claim pays when the SPX average ends low and the VIX path meets a condition that spans several windows, which is the dependence that stitching discards.
\begin{itemize}
    \item \textbf{Barriers} are fixed before any price is computed, from the single-date VIX marginals that the laws share: $B_{\mathrm{avg}}=1.10\,\bar F_V$, $B_{\mathrm{max}}=1.30\,\bar F_V$, and $B_{\mathrm{min}}=0.80\,\bar F_V$, or $0.90\,\bar F_V$ on the nine-maturity calendar.
    A pilot that sees only these marginals and two fixed couplings across dates, independence and a Gaussian copula with correlation $0.45$, requires each event to occur with probability $2\%$ to $98\%$ under both, and a barrier that fails moves to the nearest value on a predeclared grid that passes.
    The minimum barrier moves from $0.90$ to $0.80$ on the six-maturity calendars; on the nine-maturity calendar no value of the grid passes, so it stays at $0.90$, where the event occurs with probability $0.09\%$ under independence and $5.7\%$ under the Gaussian copula.
    \item \textbf{Prices are paired:} each law of a calendar is priced on $400{,}000$ paths simulated from the same random numbers, as forward values, and the differences carry paired standard errors.
    \item \textbf{A difference is material} when it exceeds both $5$ basis points of notional and four paired standard errors.
\end{itemize}

\subsection{Fit}\label{subsec:market_fit}

\Cref{tab:fit} reports the fit of the six laws.
\begin{table}[ht]
\centering
\caption{Fit of the six laws, with memory ($b=1$, the law $Q_1$) and without ($b=0$, the law $Q_0$).
Implied-volatility errors are in volatility points: the largest over the calibration strikes, and the mean over every quoted strike the grid can price, both at the worst maturity or window.
The last column is the conditional residual~\eqref{eq:dw_residual}, in units of the tolerance.}
\label{tab:fit}
\small
\begin{tabular}{@{}lcrrrrr@{}}
\toprule
 & & \multicolumn{2}{c}{SPX} & \multicolumn{2}{c}{VIX} & \\
\cmidrule(lr){3-4}\cmidrule(lr){5-6}
calendar & $b$ & calibration, max & all, mean & calibration, max & all, mean & $\mathfrak R_{\mathrm{cond}}$ \\
\midrule
non-overlapping & 1 & 0.31 & 0.93 & 0.18 & 0.78 & 0.00014 \\
 & 0 & 0.31 & 0.93 & 0.15 & 0.79 & 0.00011 \\
\addlinespace
overlapping (6) & 1 & 0.66 & 1.60 & 1.42 & 0.75 & $9.5\times10^{-5}$ \\
 & 0 & 0.63 & 1.66 & 1.18 & 0.75 & $6.0\times10^{-5}$ \\
\addlinespace
overlapping (9) & 1 & 0.67 & 1.61 & 1.18 & 0.75 & 0.00016 \\
 & 0 & 0.65 & 1.67 & 1.04 & 0.75 & 0.00019 \\
\bottomrule
\end{tabular}
\end{table}

\begin{itemize}
    \item \textbf{Calibration strikes:} the largest SPX implied-volatility error is $0.31$ to $0.67$ volatility points.
    The largest VIX error is $0.15$ to $0.18$ on the non-overlapping calendar and $1.04$ to $1.42$ on the overlapping calendars, in each case at the lowest VIX calibration strike, $0.70$ times the forward.
    \item \textbf{Bid--ask ranges:} at the calibration strikes, both laws price every quoted VIX option inside its bid--ask range.
    The SPX ranges there are about $0.15$ volatility points wide, and $Q_1$ and $Q_0$ are inside them at $74$ and $74$ of $84$ strikes on the non-overlapping calendar and at $48$ and $47$ of $84$ and $100$ and $100$ of $126$ on the overlapping ones, missing by at most $0.30$, $0.54$ and $0.55$ volatility points; the smoothed surfaces they are fitted to are inside at $83$ of $84$ and $125$ of $126$.
    \item \textbf{All quotes:} over every quoted strike the grid can price, the mean error at the worst maturity is $0.93$ to $1.67$ volatility points on the SPX and $0.75$ to $0.79$ on the VIX.
    Inside the band the SPX mean is at most $0.18$; the rest comes from the wings, which carry no calibration strikes and no conditional constraints.
    The largest SPX error over all quotes is $8.4$ to $18.0$ volatility points, in the low-strike wing at $0.71$ to $0.76$ times the forward, at $24$ or $58$ days.
    \item \textbf{Memory costs little fit:} on each calendar the two memory orders are within $0.06$ volatility points of each other in the mean SPX error and $0.005$ in the mean VIX error.
    The largest VIX error at the calibration strikes is higher with memory, by $0.02$ to $0.24$.
    \item \textbf{Log-contract constraints:} the window identity~\eqref{eq:consistency_global} misses by $1.1$ to $1.5$ volatility points at the first window, at $28$ days, and by at most $0.30$ at the others; on the non-overlapping calendar, more cycles close these misses (\cref{subsec:solver}).
    \item \textbf{Conditional constraints:} on the constrained cells, the largest residual is below $2\times10^{-4}$ of its tolerance; on the cells of paths that have left the band, where they are not imposed, the mean of $\lvert z_c\rvert$ weighted by cell mass is $6$ to $140$ (martingale) and $340$ to $612$ (dispersion).
\end{itemize}
\Cref{fig:fit} shows the fit of the two laws on the non-overlapping calendar at every maturity, over the range of calibration strikes.
Between the VIX calibration strikes the error reaches $2.3$ volatility points at $28$ days: the law's VIX smile can bend only at its $12$ atoms.

\begin{figure}[t]
\centering
\includegraphics[width=\textwidth]{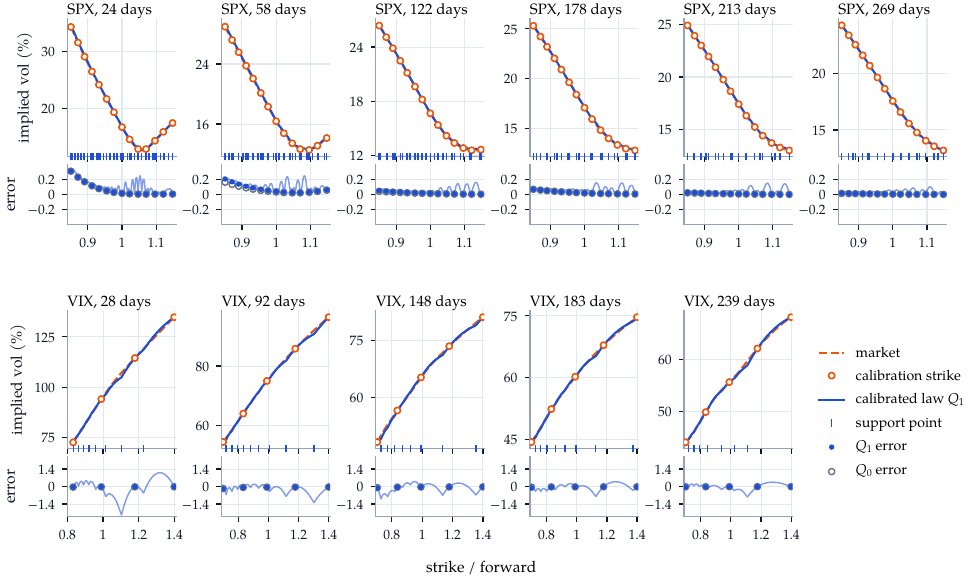}
\caption{Fit of the laws on the non-overlapping calendar at every SPX maturity (top) and VIX date (bottom), over the range of calibration strikes only; the strikes outside it carry no constraints, and their errors enter the all-quote means of \cref{tab:fit}.
Each panel shows the market smile, the smile of $Q_1$ with its support points, and below it the error in volatility points, with the errors of $Q_1$ and $Q_0$ at the calibration strikes marked.
The largest error at a calibration strike is $0.31$ (SPX) and $0.18$ (VIX) for $Q_1$.}
\label{fig:fit}
\end{figure}

\subsection{What Memory Changes}\label{subsec:market_prices}

\paragraph{Dependence.}
Memory shows in the VIX path.
The correlation of the VIX level between consecutive settlement dates is $0.341$, $0.351$ and $0.343$ under $Q_1$ on the three calendars, and $0.008$, $0.016$ and $0.016$ under $Q_0$ (\cref{tab:prices}).
On the non-overlapping calendar it is $0.039$ under the Markovization $\mathsf MQ_1$, which has the law of $Q_1$ on every step and links consecutive steps only through the SPX level and the exit indicator.
A stitched law links the VIX across dates only through the SPX level~\eqref{eq:stitched_ci}; stitching the blocks of $Q_1$ therefore keeps little of its VIX dependence.
\Cref{fig:corr} shows the correlations of the whole path.
The two laws agree, up to sampling error of $0.008$, on every pair inside one step of the timeline, as \cref{thm:feasibility_equivalence} requires.
Across steps they differ by up to $0.32$ between VIX dates, up to $0.07$ between the SPX and a later VIX, and by less than $0.004$ between SPX levels.
The SPX and the VIX at the same date are negatively correlated in both, $-0.14$ to $-0.27$ under $Q_1$; the references draw the VIX independently of the SPX, so this dependence comes from the calibration.
\Cref{fig:vixpairs} in \cref{app:discretization} shows the joint law of consecutive VIX levels: under $Q_1$, high follows high and low follows low; under its Markovization $\mathsf MQ_1$, the joint law is close to independent.

\begin{figure}[t]
\centering
\includegraphics[width=\textwidth]{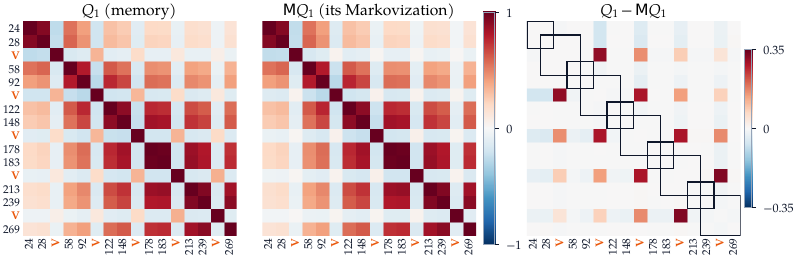}
\caption{Correlation of the path on the non-overlapping calendar: the SPX at every node (by day) and the VIX at every VIX date (V, after the SPX of its date), from the $400{,}000$ paths behind \cref{tab:memory}.
The boxes in the difference mark the steps of the timeline, whose laws the Markovization keeps.}
\label{fig:corr}
\end{figure}

\paragraph{Prices on the non-overlapping calendar.}
\Cref{tab:memory} prices the claims under $Q_1$, $Q_0$ and $\mathsf MQ_1$ and splits each difference by~\eqref{eq:memory_split}.
\begin{table}[ht]
\centering
\caption{Prices on the non-overlapping calendar: forward values in percent of notional, and differences in basis points of notional with paired standard errors.
The difference $P_{Q_1}-P_{Q_0}$ splits by~\eqref{eq:memory_split} into the memory term, $P_{Q_1}-P_{\mathsf MQ_1}$, and the block-selection term, $P_{\mathsf MQ_1}-P_{Q_0}$.
Barriers are in multiples of $\bar F_V$.}
\label{tab:memory}
\small
\begin{tabular}{@{}lcrrrrrr@{}}
\toprule
claim & barrier & $Q_1$ & $Q_0$ & $\mathsf MQ_1$ & $Q_1-Q_0$ & memory & selection \\
\midrule
maximum & 1.30 & 2.2790 & 2.5001 & 2.6415 & $-22.12$ (0.84) & $-36.25$ (1.13) & $+14.14$ (1.10) \\
minimum & 0.80 & 1.1522 & 0.4306 & 0.6973 & $+72.16$ (0.78) & $+45.49$ (0.95) & $+26.67$ (0.78) \\
\addlinespace
smooth & 1.05--1.20 & 7.6354 & 7.5221 & 8.2163 & $+11.33$ (2.87) & $-58.10$ (3.89) & $+69.43$ (3.77) \\
average & 1.10 & 1.8964 & 1.8510 & 2.1060 & $+4.54$ (0.86) & $-20.95$ (1.14) & $+25.50$ (1.10) \\
\bottomrule
\end{tabular}
\end{table}

\begin{itemize}
    \item \textbf{Maximum and minimum claims:} memory makes the maximum claim cheaper, by $22.12$ basis points of notional, and the minimum claim dearer, by $72.16$.
    This is what positive dependence of the VIX across windows does: fewer independent chances to cross a high barrier, and more paths that stay above a low one.
    \item \textbf{The memory term} is $-36.25$ for the maximum claim and $+45.49$ for the minimum claim.
    The change of the blocks raises both prices, by $14.14$ and $26.67$, which offsets part of the first memory term and adds to the second.
    With a Markovization in the SPX level alone (\cref{subsec:discretization}), the memory terms are $-43.31$ and $+43.35$ and the block-selection terms $+21.20$ and $+28.81$; the dependence is the larger part of both differences either way.
    \item \textbf{Smooth and average claims:} their memory terms are negative too, $-58.10$ and $-20.95$, but the change of the blocks raises them by about as much, $+69.43$ and $+25.50$, so their differences, $+11.33$ and $+4.54$, are small sums of large opposite terms and are not material.
    Their memory terms stay large and negative on every grid, while their net differences move with the grid (\cref{app:discretization}); the comparison therefore rests on the maximum and minimum claims.
\end{itemize}

\paragraph{Prices on all three calendars.}
\Cref{tab:prices} gives $P_{Q_1}-P_{Q_0}$ for the maximum and minimum claims on every calendar; \cref{tab:allclaims} gives all four claims.
Under memory the maximum claim is $18.58$ to $22.94$ basis points cheaper and the minimum claim $19.12$ to $72.16$ basis points dearer, and each of these differences is material.
On the nine-maturity calendar the minimum claim is rare without memory: it is worth $1.0$ basis points under $Q_0$ and $20.1$ under $Q_1$.
With the barrier at $0.80$ or $1.00$, the other values of its grid, the difference is $+45.12$ or $+9.78$ basis points, both material.
On the overlapping calendars we report the differences without the memory split, which applies when every seam is a single SPX node (\cref{sec:reference}).

\begin{table}[ht]
\centering
\caption{$P_{Q_1}-P_{Q_0}$ for the maximum and minimum claims on each calendar, in basis points of notional, with paired standard errors, and the correlation of the VIX level between consecutive settlement dates under each law.}
\label{tab:prices}
\small
\begin{tabular}{@{}lrrr@{}}
\toprule
 & non-overlapping & overlapping (6) & overlapping (9) \\
\midrule
maximum & $-22.12$ (0.84) & $-18.58$ (0.70) & $-22.94$ (0.85) \\
minimum & $+72.16$ (0.78) & $+66.73$ (0.70) & $+19.12$ (0.41) \\
\addlinespace
VIX lag-one correlation, $Q_1$ & 0.341 & 0.351 & 0.343 \\
VIX lag-one correlation, $Q_0$ & 0.008 & 0.016 & 0.016 \\
\bottomrule
\end{tabular}
\end{table}

\paragraph{Memory strength.}
\Cref{tab:dial} recalibrates the law with memory on the non-overlapping calendar at $\varphi=0$, $0.3$ and $0.8$, with everything else fixed, and prices each against the same $Q_0$.
\begin{table}[ht]
\centering
\caption{Memory strength on the non-overlapping calendar: the law with memory at four memory strengths $\varphi$, the baseline law at $\varphi=0.6$ and recalibrations at the other three, each priced against the same $Q_0$.
Lag-one correlations of the discretized latent state and of the VIX level under $Q_1$; $P_{Q_1}-P_{Q_0}$ in basis points of notional with paired standard errors, and the memory term $P_{Q_1}-P_{\mathsf MQ_1}$.
At $\varphi=0.8$ the Gauss--Hermite grid with $n_Z=7$ gives a latent correlation of $0.813$.}
\label{tab:dial}
\small
\begin{tabular}{@{}lrrrrrr@{}}
\toprule
 & \multicolumn{2}{c}{lag-one correlation} & \multicolumn{2}{c}{$P_{Q_1}-P_{Q_0}$} & \multicolumn{2}{c}{memory} \\
\cmidrule(lr){2-3}\cmidrule(lr){4-5}\cmidrule(lr){6-7}
$\varphi$ & latent & VIX & maximum & minimum & maximum & minimum \\
\midrule
0.0 & 0.000 & 0.008 & $+0.00$ (0.00) & $+0.00$ (0.00) & $-0.30$ & $+0.17$ \\
0.3 & 0.300 & 0.166 & $-8.09$ (0.59) & $+27.70$ (0.49) & $-16.29$ & $+17.07$ \\
0.6 & 0.600 & 0.341 & $-22.12$ (0.84) & $+72.16$ (0.78) & $-36.25$ & $+45.49$ \\
0.8 & 0.813 & 0.483 & $-43.88$ (1.04) & $+112.42$ (0.98) & $-56.56$ & $+73.66$ \\
\bottomrule
\end{tabular}
\end{table}

\begin{itemize}
    \item \textbf{At $\varphi=0$} the latent states are independent, the reference is the Markov one, and the calibrated law and its prices are those of $Q_0$.
    Its memory terms, $-0.30$ and $+0.17$ basis points with standard errors $1.11$ and $0.65$, vanish within sampling error: $Q_0$ prices these claims as its Markovization does.
    \item \textbf{The differences grow with $\varphi$:} $-8.09$, $-22.12$ and $-43.88$ basis points for the maximum claim and $+27.70$, $+72.16$ and $+112.42$ for the minimum claim at $\varphi=0.3$, $0.6$ and $0.8$; the memory terms grow the same way.
    \item \textbf{The fit stays the same:} the mean SPX error is $0.93$ to $0.94$ volatility points at every $\varphi$, the mean VIX error $0.78$ to $0.79$, and the largest errors at the calibration strikes stay at $0.31$ (SPX) and $0.15$ to $0.18$ (VIX).
\end{itemize}
The vanilla quotes therefore do not determine these prices: they fix the law of each date, every law of this family fits them, and these claims depend on the joint law across dates, which changes with $\varphi$ (\cref{fig:dial}).

\begin{figure}[t]
\centering
\includegraphics[width=\textwidth]{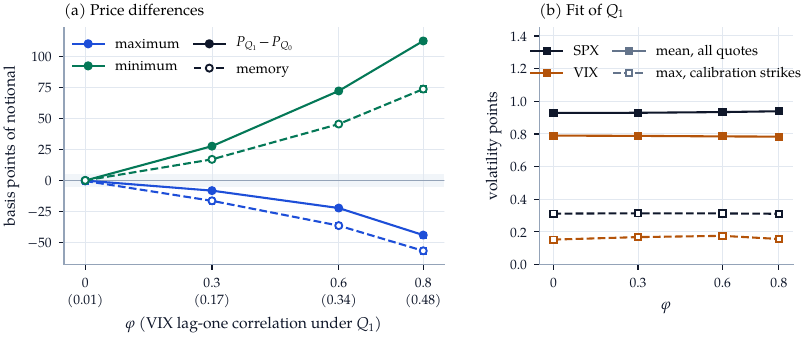}
\caption{Memory strength on the non-overlapping calendar, for the laws of \cref{tab:dial}.
(a) $P_{Q_1}-P_{Q_0}$ (solid) and the memory term $P_{Q_1}-P_{\mathsf MQ_1}$ (dashed) of the maximum and minimum claims, with bars of two paired standard errors; in parentheses, the VIX lag-one correlation under $Q_1$.
The grey band marks differences below $5$ basis points, which are not material.
(b) The fit of $Q_1$: the mean implied-volatility error over all quotes and the largest at the calibration strikes, at the worst maturity, as in \cref{tab:fit}.}
\label{fig:dial}
\end{figure}

\subsection{Robustness}\label{subsec:robustness}

\paragraph{Grids and sampling.}
We recalibrated and repriced the pair of the non-overlapping calendar on finer grids, up to $n_Z=15$ equal-probability and $n_Z=11$ Gauss--Hermite latent nodes, $n_W=24$ and $n_S=80$, with $150$ cycles, and with up to $10{,}000{,}000$ paths (\cref{app:discretization}).
On the baseline law's latent grid and finer ones, Gauss--Hermite nodes from $n_Z=7$ and equal-probability nodes from $n_Z=12$, which need about twice as many nodes (\cref{subsec:discretization}), the maximum claim's difference ranges from $-25.26$ to $-20.06$ basis points and the minimum claim's from $+71.27$ to $+76.02$, with memory terms from $-39.35$ to $-33.95$ and from $+44.21$ to $+46.71$.
At $10{,}000{,}000$ paths the standard errors are $0.17$ basis points, so the spread comes from the grids, not from sampling.
Coarser latent grids carry a different memory strength (\cref{subsec:discretization}) and so change the size of the differences, not their signs: \cref{fig:variants} in \cref{app:discretization} shows both differences for every variant, coarse grids and the other calendars included, each with the same sign as for the baseline pair, and material.

\paragraph{The band.}
On all three calendars, paths that leave the band before the last step carry $85$ to $98\%$ of the maximum and minimum claims' prices under $Q_1$.
On the non-overlapping calendar we therefore recalibrated the pair with both conditional families imposed on every path, so that (C2)--(C3) also hold after the SPX leaves the band, and with one family on every path and the other only in the band.
Each variant keeps the signs and materiality of both differences: the maximum claim's ranges from $-18.58$ to $-16.82$ basis points and the minimum claim's from $+85.01$ to $+99.28$ (\cref{tab:band}).

\paragraph{Latent discretization.}
The latent grid sets how much memory the reference carries (\cref{subsec:discretization}), and the prices follow it.
At $n_Z=3$ the maximum claim's difference is $-6.75$ to $-9.87$ basis points on the equal-probability grid and $-33.24$ to $-40.28$ on the Gauss--Hermite grid, and the values at $n_Z=7$ lie between; the minimum claim's difference at $n_Z=7$ is larger than on either grid at $n_Z=3$ on the six-maturity calendars and equals the Gauss--Hermite value within sampling error on the nine-maturity calendar (\cref{tab:crosscal}).
On the non-overlapping calendar, equal-probability nodes give at $n_Z=12$ to $15$ the prices that Gauss--Hermite nodes give at $n_Z=7$ (\cref{tab:convergence}).

\section{Conclusion and Future Work}\label{sec:conclusion}

We moved multi-maturity SPX--VIX calibration from one month at a time to the whole path and separated feasibility from identification.
Local and global feasibility are equivalent, and block-preserving Markovization maps every global law to a stitched law with the same monthly calibrations; but the stitched laws can form a strict subset of the global ones, and under the standard Markov reference relative entropy selects the stitched completion.
Dependence across periods must therefore be put in deliberately.
We put it in the reference, through a latent memory state that keeps every VIX marginal exact, and calibrated the resulting law on a junction tree whose size grows linearly with the number of maturities.
On market data with up to nine SPX maturities and eight VIX dates, laws with and without memory fit the same quotes about equally well but price VIX-conditioned SPX Asian claims differently.
On every calendar, memory makes a claim on the maximum VIX $19$ to $23$ basis points of notional cheaper and a claim on the minimum VIX $19$ to $72$ basis points dearer.
Where the memory split exists, the dependence is the larger part of both differences; both grow steadily with the memory strength.
On the non-overlapping calendar, both differences also keep their signs when the martingale and dispersion conditions are imposed on every path, not only while the SPX stays in a band around the forward.

Several directions seem worth exploring.
The laws computed here are a few of the many that fit the same quotes, and price bounds over all of them, or over all gluings of given monthly blocks as in~\eqref{eq:robust_bounds}, could show how much of a claim's price the quotes leave open.
The memory strength $\varphi$ is fixed here; quoted instruments whose payoffs depend on several months jointly might help to choose it, much as path-dependent optimal transport calibrates volatility models to the prices of path-dependent derivatives~\citep{guo2021path}.
Holding (C2)--(C3) on every path without loosening the fit may require choosing the SPX wings, which carry no quotes but part of the VIX variance (\cref{tab:gap}), consistently with the VIX surface.
Higher memory orders, SPX kernels that read the latent state, so that the SPX also responds to the VIX history, and an exact reference for overlapping VIX windows would be natural further steps.
Risk sensitivities of a calibrated law with memory could be computed by perturbing the transport problem, as \citet{che2026spx} do for a calibrated coupling.

\clearpage
\appendix
\section*{Appendix}

\section{Proofs of the Structural and Variational Results}\label[appendix]{app:proofs}

\subsection{Feasibility equivalence and Markovization}\label[appendix]{app:proof_feasibility_equivalence}

\begin{proof}[Proof of \cref{thm:feasibility_equivalence}]
Let $\mu\in\cP_{\mathrm{full}}$.
Since $\sigma(S_i,V_i)\subseteq\cF_i$, the tower property gives
\begin{align*}
\Embb^\mu[S_{i+1}\mid S_i,V_i]
&=\Embb^\mu[\Embb^\mu[S_{i+1}\mid\cF_i]\mid S_i,V_i]=S_i,\\
\Embb^\mu[L(S_{i+1}/S_i)\mid S_i,V_i]
&=\Embb^\mu[\Embb^\mu[L(S_{i+1}/S_i)\mid\cF_i]\mid S_i,V_i]=V_i^2.
\end{align*}
Thus $\nu_i^\mu:=\Law_\mu(S_i,V_i,S_{i+1})$ belongs to $\cP_i$ for every $i$.

The product in~\eqref{eq:markovization} preserves $\nu_1^\mu$.
Inductively, if it has the same $S_i$ marginal as $\mu$, adjoining the kernel $\mu(dv_i,ds_{i+1}\mid s_i)$ gives the same $(S_i,V_i,S_{i+1})$ law $\nu_i^\mu$; hence it also preserves the next $S_{i+1}$ marginal.
Every adjacent triple law is therefore preserved.
Moreover, the product factorization gives
\[
\Law_{\mathsf M\mu}(S_{i+1}\mid\cF_i)
=\Law_{\nu_i^\mu}(S_{i+1}\mid S_i,V_i),
\]
so the two local conditional identities imply (C2)--(C3) under $\mathsf M\mu$.
Thus $\mathsf M\mu\in\cP_{\mathrm{stitch}}\subseteq\cP_{\mathrm{full}}$.

Finally, if every $\cP_i$ is nonempty, choose $\nu_i\in\cP_i$ and glue them by the product in~\eqref{eq:stitched_class}; their prescribed shared SPX marginals agree, and the preceding argument proves that the resulting law is globally feasible.
The reverse implications follow from the block preservation just established.
\end{proof}

\subsection{Strict inclusion and non-identification}\label[appendix]{app:proof_strict_inclusion}

\begin{proof}[Proof of \cref{thm:strict_inclusion}]
The inclusion follows from \cref{thm:feasibility_equivalence}.
For strictness, take $m=3$, set $S_0=S_1=100$, and let $V_1$ be equally likely to equal $v_L=0.15$ or $v_H=0.55$.
For $(v,a)=(v_L,8)$ or $(v_H,30)$, set
\[
q(v,a):=\frac{v^2}{-(2/\tau)\log(1-(a/100)^2)}\in(0,1/2)
\]
and, conditional on $V_1=v$, assign probabilities $q(v,a)$, $1-2q(v,a)$, and $q(v,a)$ to $S_2=100-a$, $100$, and $100+a$.
Symmetry gives $\Embb[S_2\mid S_1,V_1]=S_1$, while
\[
q(v,a)\left[-\frac{2}{\tau}\log(1-a/100)-\frac{2}{\tau}\log(1+a/100)\right]=v^2,
\]
which gives the first dispersion identity.

Set $V_2=0.20$ when $V_1=v_L$ and $V_2=0.60$ when $V_1=v_H$ and, conditionally on the history, let
\[
S_3=S_2(1+\varepsilon\delta(V_2)),
\qquad
\delta(v):=\sqrt{1-e^{-v^2\tau}},
\]
where $\varepsilon$ is conditionally uniform on $\{-1,1\}$.
Then $\Embb[S_3\mid\cF_2]=S_2$ and
\[
\Embb[L(S_3/S_2)\mid\cF_2]
=-\frac{1}{\tau}\log(1-\delta(V_2)^2)=V_2^2,
\]
so the resulting finite-support law $\mu$ belongs to $\cP_{\mathrm{full}}$.

Both $V_1$ branches reach $S_2=100$ with positive probability, but $V_2$ identifies the branch there.
Hence $V_2\not\perp V_1\mid S_2$, so~\eqref{eq:stitched_ci} fails and $\mu\notin\cP_{\mathrm{stitch}}$.
Its Markovization belongs to $\cP_{\mathrm{stitch}}$ and preserves both adjacent triple laws by \cref{thm:feasibility_equivalence}.
Because the two finite laws differ, the indicator of any atom on which their masses differ is a bounded separating payoff.
Finally, equality $\mu=\mathsf M\mu$ is equivalent to equality of the full-history kernels and the $S_i$-conditional kernels in~\eqref{eq:markovization}, which is exactly~\eqref{eq:stitched_ci}.
For $m>3$, append degenerate feasible periods.
\end{proof}

\subsection{Information and entropy identities}\label[appendix]{app:proof_information_loss}

\begin{proof}[Proof of \cref{prop:information_loss}]
Disintegrate $\mu$ chronologically as
\[
\mu(d\mathbf s,d\mathbf v)
=\nu_1^\mu(ds_1,dv_1,ds_2)
\prod_{i=2}^{m-1}\mu(dv_i,ds_{i+1}\mid H_{i-1},s_i),
\]
whereas~\eqref{eq:markovization} replaces each displayed continuation kernel by $\mu(dv_i,ds_{i+1}\mid s_i)$.
The chain rule for relative entropy therefore gives
\begin{align*}
D_{\mathrm{KL}}(\mu\Vert\mathsf M\mu)
&=\sum_{i=2}^{m-1}
\Embb^\mu\!\left[
\log\frac{d\mu((V_i,S_{i+1})\mid H_{i-1},S_i)}
{d\mu((V_i,S_{i+1})\mid S_i)}
\right]\\
&=\sum_{i=2}^{m-1}I_\mu((V_i,S_{i+1});H_{i-1}\mid S_i).
\end{align*}
Each summand is nonnegative and vanishes exactly when the corresponding conditional-independence relation~\eqref{eq:stitched_ci} holds.
\end{proof}

\subsection{Duality and existence}\label[appendix]{app:proof_duality}

\begin{proof}[Proof of \cref{thm:duality}]
If $\mu\in\cP_{\mathrm{full}}$, then every bounded predictable gain in~\eqref{eq:portfolio_value} has zero expectation by (C2)--(C3), while (C1) prices the static terms at their initial costs.
Thus $\Embb^\mu[\Pi]=\mathrm{Cost}(\Pi)$ for every $\Pi\in\cU_{\mathrm{full}}$, ruling out a pointwise nonnegative portfolio of negative cost.

Conversely, for a law satisfying (C1), the identity
\[
\Embb^\mu[\Delta_{S,i}(\mathbf S^i,\mathbf V^i)(S_{i+1}-S_i)]=0
\]
for every bounded $\cF_i$-measurable $\Delta_{S,i}$ is equivalent to (C2); the analogous identity with $L(S_{i+1}/S_i)-V_i^2$ is equivalent to (C3).
The dynamic multipliers therefore dualize exactly the two conditional families.
On a finite state space, infeasibility of the resulting linear system is equivalent by Farkas' lemma to the existence of multipliers whose combined payoff is nonnegative on every state and whose initial cost is negative.
This is the finite-dimensional form of the martingale-transport separation argument; continuum analogues require the additional hypotheses discussed in \cref{subsec:duality}, as in \citet{Beiglbck2013,Beiglbck2017} and the SPX--VIX formulation of \citet{guyon2020joint}.
\end{proof}

\subsection{Entropic existence and the reference measure}\label[appendix]{app:proof_entropic_af}

\begin{proof}[Proof of \cref{thm:entropic_af}]
Take a minimizing sequence in the sublevel set $\{\mu:D_{\mathrm{KL}}(\mu\Vert\bar\mu)\le D_{\mathrm{KL}}(\mu_0\Vert\bar\mu)\}$, which is weakly compact because relative entropy with respect to a probability measure on a Polish space has weakly compact sublevel sets.
It has a weakly convergent subsequence; weak closedness of $\cP_{\mathrm{full}}$ keeps the limit feasible, and lower semicontinuity of relative entropy makes that limit a minimizer.
Relative entropy is strictly convex on laws dominated by $\bar\mu$, while $\cP_{\mathrm{full}}$ is convex, so two distinct finite-entropy minimizers cannot exist.
The resulting feasible law rules out global arbitrage by \cref{thm:duality}.
On a finite allowed support, the simplex is compact, the affine constraint set is closed, and a strictly positive reference dominates every feasible vector; hence nonemptiness supplies all of the stated conditions.
\end{proof}

\begin{proof}[Proof of \cref{prop:ref_props}]
The iterated product in~\eqref{eq:complete_ref} is a probability measure by construction.
The lognormal kernel $T_i$ satisfies
\[
\Embb[S_{i+1}\mid S_i=s_i,V_i=v_i]=s_i
\]
and, since $\log(S_{i+1}/S_i)=v_i\sqrt\tau G-v_i^2\tau/2$,
\[
\Embb[L(S_{i+1}/S_i)\mid S_i=s_i,V_i=v_i]=v_i^2.
\]
The product construction makes these identities valid conditional on the full history and draws each $V_i$ from $\mu_{V_i}$ independently of the preceding state; it also starts from $\mu_{S_1}$.
Its support is the product-kernel support generated by these factors, not necessarily the entire ambient space.
On a finite allowed grid, strictly positive reference weights dominate every law on that grid.
In the continuum, neither topological full support nor equality of supports implies absolute continuity, which is why \cref{thm:entropic_af} assumes a feasible finite-entropy law separately.
\end{proof}

\subsection{KL selection of the stitched completion}\label[appendix]{app:proof_kl_markovization}

\begin{proof}[Proof of \cref{prop:kl_markovization}]
Apply the relative-entropy chain rule to the full-history kernels of $\mu$ and the $S_i$-conditional kernels of $\bar\mu$.
For each $i\ge2$, insert the intermediate kernel $\mu(dv_i,ds_{i+1}\mid s_i)$ into the log-density ratio.
The terms comparing the full-history kernel with this intermediate kernel sum to $D_{\mathrm{KL}}(\mu\Vert\mathsf M\mu)$ by \cref{prop:information_loss}.
The remaining terms depend only on the adjacent block laws; because $\mu$ and $\mathsf M\mu$ share those laws, together with the first block, they sum to $D_{\mathrm{KL}}(\mathsf M\mu\Vert\bar\mu)$.
This proves~\eqref{eq:kl_markovization}.
Since $\mathsf M\mu$ is feasible by \cref{thm:feasibility_equivalence}, Markovization weakly decreases the objective, strictly so whenever $D_{\mathrm{KL}}(\mu\Vert\mathsf M\mu)>0$.
Therefore every finite-entropy minimizer is stitched.
\end{proof}

\subsection{Compatibility under bid--ask bands}\label[appendix]{app:proof_band_compatibility}

\begin{proof}[Proof of \cref{prop:band_compatibility}]
If $\mu\in\cP_{\mathrm{full}}^{\mathrm{band}}$, set $\alpha_i=\Law_\mu(S_i)$ and $\beta_i=\Law_\mu(V_i)$.
The adjacent triple law $\Law_\mu(S_i,V_i,S_{i+1})$ is locally feasible by the tower property, so $(\alpha_i,\beta_i,\alpha_{i+1})\in\mathfrak R_i$ for every $i$.
The sequence is common across all blocks because each $S_i$ is one random variable under $\mu$.

Conversely, suppose one sequence $(\alpha_1,\beta_1,\ldots,\beta_{m-1},\alpha_m)$ belongs to all adjacent relations.
Choose $\nu_i\in\cP(\alpha_i,\beta_i,\alpha_{i+1})$ for each $i$ and disintegrate each later block with respect to its first SPX coordinate.
The product
\[
\nu_1(ds_1,dv_1,ds_2)
\prod_{i=2}^{m-1}\nu_i(dv_i,ds_{i+1}\mid s_i)
\]
is well-defined because adjacent blocks share $\alpha_i$.
It preserves every selected marginal and, by the same conditional-expectation argument as in \cref{app:proof_feasibility_equivalence}, satisfies (C2)--(C3).
It therefore belongs to $\cP_{\mathrm{full}}^{\mathrm{band}}$.
Separate nonemptiness of the $\mathfrak R_i$ is not enough.
Take $m=3$ and quotes that pin $S_1=S_0$, $V_1=V_2=v$ and the law $\alpha_3$ of $S_0\exp\bigl(v'\sqrt\tau\,G+v\sqrt\tau\,G'-(v'^2+v^2)\tau/2\bigr)$, with $G,G'$ independent standard normal and $v'\neq v$, while the band at $T_2$ admits the laws $\alpha_2$ and $\alpha_2'$ of $S_0\exp(\sigma\sqrt\tau\,G-\sigma^2\tau/2)$ with $\sigma=v$ and $\sigma=v'$.
The kernel~\eqref{eq:lognormal_kernel} gives $(\delta_{S_0},\delta_v,\alpha_2)\in\mathfrak R_1$ and $(\alpha_2',\delta_v,\alpha_3)\in\mathfrak R_2$.
A common selection $\alpha$ would need, by~\eqref{eq:consistency_global} for both blocks, $\int L\,d\alpha=L(S_0)+v^2$ and $\int L\,d\alpha=\int L\,d\alpha_3-v^2=L(S_0)+v'^2$, which is impossible.
\end{proof}

\subsection{Merged-timeline reference}\label[appendix]{app:proof_merged_reference}

\begin{proof}[Proof of \cref{prop:merged_reference}]
Starting from the prescribed law of $S_{t_1}$, successively adjoining the stated VIX draws and SPX transition kernels defines a probability law by finite iterated disintegration.
Its $V_i$ marginals are the prescribed ones by construction.
Partition a VIX window that overlaps no other into its steps, of lengths $\Delta t_k$, and let $d_{i,k}$ be the corresponding forward ratios.
Conditional on $\cF_{T_i^V}$ and $V_i$, define independent increments by
\[
\frac{S_{t_{k+1}}}{d_{i,k}S_{t_k}}
=\exp\!\left(V_i\sqrt{\Delta t_k}\,G_k-\tfrac12V_i^2\Delta t_k\right),
\qquad G_k\sim\mathcal N(0,1).
\]
Each increment has conditional mean one, so their product is a forward-adjusted martingale over the window.
Since $\sum_k\Delta t_k=\tau_i$ and the Gaussian terms have conditional mean zero,
\[
\Embb\!\left[-\frac{2}{\tau_i}\sum_k
\log\!\left(\frac{S_{t_{k+1}}}{d_{i,k}S_{t_k}}\right)
\middle|\cF_{T_i^V}\right]
=\frac{1}{\tau_i}\sum_k V_i^2\Delta t_k
=V_i^2.
\]
The sum of logs is the log of the forward-adjusted endpoint ratio.
Non-overlap ensures that no subinterval is assigned two different VIX levels, while the stipulated kernels on uncovered intervals preserve the martingale property.
The complete chronological product therefore has all the properties claimed.
\end{proof}

\subsection{Missing intermediate SPX marginals}\label[appendix]{app:proof_stitch_illposed}

\begin{proof}[Proof of \cref{prop:stitch_illposed}]
A standard three-variable block beginning at $T_i^V$ takes the initial SPX marginal as part of its input.
When $T_i^V\notin\mathcal T_S$, the supplied SPX option surfaces specify no such marginal.
Consequently the block is not determined by the stated market inputs: one must either select a marginal by an additional interpolation rule or optimize over it as an endogenous variable in an enlarged state space.
\end{proof}

\section{Discretization and Robustness}\label[appendix]{app:discretization}

This appendix shows the discretization of \cref{subsec:discretization}, gives the runs behind the robustness statements of \cref{subsec:robustness} and adds figures and tables to \cref{subsec:market_fit,subsec:market_prices}.
In the robustness runs, every pair on the non-overlapping calendar is recalibrated with everything else fixed and priced on the same barriers.

\paragraph{Support points.}
\Cref{fig:support} shows the discretization of \cref{subsec:discretization} on the non-overlapping calendar.

\begin{figure}[t]
\centering
\includegraphics[width=\textwidth]{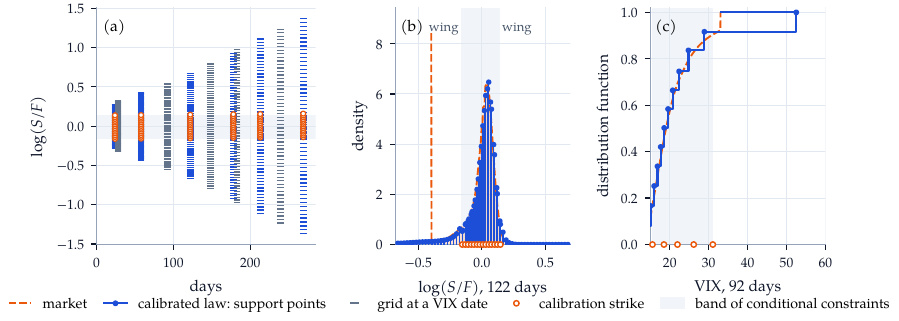}
\caption{Support points and calibration strikes on the non-overlapping calendar.
(a) The SPX grid at every node of the merged timeline, in $\log(S/F)$: quoted maturities in blue, VIX dates in grey, calibration strikes as open circles, and the band shaded; outside it are the wings.
(b) The calibrated law at $122$ days against the market density: $70$ support points, $25$ of them inside the band, which carries $87.8\%$ of the law's mass.
(c) The law's $12$ VIX atoms at $92$ days, as a distribution function, against the market's, with the five calibration strikes.}
\label{fig:support}
\end{figure}

\paragraph{Latent dependence.}
\Cref{tab:latent} gives the lag-one correlation of the discretized latent state.
\begin{itemize}
    \item \textbf{Equal-probability nodes} under-couple at every $n_Z$ and approach $\varphi$ slowly.
    \item \textbf{Gauss--Hermite nodes} over-couple at $n_Z=3$, most for large $\varphi$, and are within $0.002$ of $\varphi$ from $n_Z=5$ for $\varphi=0.3$ and $0.6$ and from $n_Z=12$ for $\varphi=0.8$.
    \item \textbf{The emission grid} accounts for the remaining gap: under the reference with $n_W=12$, the lag-one correlation of the emission is $0.467$ at $n_Z=7$ and $0.477$ at $n_Z=11$ on Gauss--Hermite nodes, against $\rho_e^2\varphi=0.486$ (\cref{prop:finite_memory}).
\end{itemize}

\begin{table}[ht]
\centering
\caption{Lag-one correlation of the discretized latent state, for three memory strengths $\varphi$ and two latent grids, against the number of latent nodes $n_Z$.}
\label{tab:latent}
\small
\begin{tabular}{@{}llrrrrrr@{}}
\toprule
$\varphi$ & latent grid & $n_Z=3$ & $5$ & $7$ & $9$ & $12$ & $20$ \\
\midrule
0.3 & equal-probability & 0.200 & 0.241 & 0.258 & 0.267 & 0.275 & 0.285 \\
 & Gauss--Hermite & 0.301 & 0.300 & 0.300 & 0.300 & 0.300 & 0.300 \\
\addlinespace
0.6 & equal-probability & 0.488 & 0.538 & 0.557 & 0.567 & 0.575 & 0.585 \\
 & Gauss--Hermite & 0.639 & 0.602 & 0.600 & 0.600 & 0.600 & 0.600 \\
\addlinespace
0.8 & equal-probability & 0.749 & 0.768 & 0.777 & 0.782 & 0.786 & 0.792 \\
 & Gauss--Hermite & 0.932 & 0.848 & 0.813 & 0.803 & 0.800 & 0.800 \\
\bottomrule
\end{tabular}
\end{table}

\paragraph{Prices as the latent grid is refined.}
\Cref{tab:convergence} recalibrates the pair of the non-overlapping calendar on both latent grids, up to $n_Z=15$ and $11$.
\begin{itemize}
    \item \textbf{Maximum and minimum claims:} the differences keep their signs and are material on both grids at every $n_Z$.
    \item \textbf{Equal-probability nodes} approach from below: at $n_Z=12$ to $15$ they give the maximum- and minimum-claim prices and the VIX correlation that Gauss--Hermite nodes give at $n_Z=7$.
    Gauss--Hermite nodes therefore reach a given memory strength with fewer nodes.
    \item \textbf{Neither grid has settled:} with $10{,}000{,}000$ paths, from $n_Z=7$ to $11$ on Gauss--Hermite nodes the maximum claim's difference moves from $-22.33$ to $-23.94$ basis points and the minimum claim's from $+71.66$ to $+75.59$ (\cref{tab:refinements}).
    \item \textbf{Smooth and average claims:} their differences change with $n_Z$ on both grids, and the smooth claim's changes sign; as in \cref{subsec:market_prices}, the comparison rests on the maximum and minimum claims.
\end{itemize}

\begin{table}[ht]
\centering
\caption{The non-overlapping calendar on two latent grids against the number of latent nodes $n_Z$, with $n_W=12$.
Lag-one correlations of the discretized latent state ($\varphi=0.6$) and of the VIX under $Q_1$; $P_{Q_1}-P_{Q_0}$ of the four claims in basis points of notional (standard errors about $0.8$, and $2.8$ for the smooth claim).
The VIX correlations and the prices use $400{,}000$ paths.
The Gauss--Hermite row at $n_Z=7$ is the law of \cref{sec:experiments}.}
\label{tab:convergence}
\small
\begin{tabular}{@{}lcrrrrrr@{}}
\toprule
 & & \multicolumn{2}{c}{lag-one correlation} & \multicolumn{4}{c}{$P_{Q_1}-P_{Q_0}$} \\
\cmidrule(lr){3-4}\cmidrule(lr){5-8}
latent grid & $n_Z$ & latent & VIX & maximum & minimum & smooth & average \\
\midrule
equal-probability & 3 & 0.488 & 0.207 & $-9.02$ & $+58.85$ & $+47.93$ & $+10.74$ \\
 & 5 & 0.538 & 0.269 & $-14.63$ & $+66.24$ & $+45.70$ & $+10.64$ \\
 & 7 & 0.557 & 0.300 & $-18.56$ & $+69.28$ & $+28.44$ & $+7.57$ \\
 & 9 & 0.567 & 0.317 & $-20.96$ & $+71.13$ & $+21.14$ & $+5.68$ \\
 & 12 & 0.575 & 0.333 & $-22.45$ & $+72.36$ & $+14.87$ & $+4.70$ \\
 & 15 & 0.580 & 0.340 & $-22.85$ & $+73.49$ & $+13.19$ & $+4.34$ \\
\midrule
Gauss--Hermite & 3 & 0.639 & 0.393 & $-39.78$ & $+56.82$ & $-40.13$ & $-5.78$ \\
 & 5 & 0.602 & 0.331 & $-25.22$ & $+66.23$ & $+13.58$ & $+4.42$ \\
 & 7 & 0.600 & 0.341 & $-22.12$ & $+72.16$ & $+11.33$ & $+4.54$ \\
 & 9 & 0.600 & 0.357 & $-24.66$ & $+73.69$ & $-1.30$ & $+1.29$ \\
 & 11 & 0.600 & 0.367 & $-25.26$ & $+74.19$ & $-7.60$ & $+0.04$ \\
\bottomrule
\end{tabular}
\end{table}

\paragraph{Across calendars.}
\Cref{tab:crosscal} gives the maximum and minimum claims on every calendar at $n_Z=3$ on both latent grids and at $n_Z=7$.
On every calendar the maximum claim's difference at $n_Z=7$ lies between the two values at $n_Z=3$; the minimum claim's is larger than either on the six-maturity calendars, and on the nine-maturity calendar it equals the Gauss--Hermite value at $n_Z=3$ within sampling error.

\begin{table}[ht]
\centering
\caption{$P_{Q_1}-P_{Q_0}$ for the maximum and minimum claims on each calendar and latent grid, in basis points of notional, with paired standard errors ($400{,}000$ paths).}
\label{tab:crosscal}
\small
\begin{tabular}{@{}llrrr@{}}
\toprule
 & & \multicolumn{2}{c}{$n_Z=3$} & $n_Z=7$ \\
\cmidrule(lr){3-4}\cmidrule(lr){5-5}
calendar & claim & equal-probability & Gauss--Hermite & Gauss--Hermite \\
\midrule
non-overlapping & maximum & $-9.02$ (0.73) & $-39.78$ (0.91) & $-22.12$ (0.84) \\
 & minimum & $+58.85$ (0.69) & $+56.82$ (0.76) & $+72.16$ (0.78) \\
\addlinespace
overlapping (6) & maximum & $-6.75$ (0.61) & $-33.24$ (0.75) & $-18.58$ (0.70) \\
 & minimum & $+54.70$ (0.62) & $+55.05$ (0.69) & $+66.73$ (0.70) \\
\addlinespace
overlapping (9) & maximum & $-9.87$ (0.73) & $-40.28$ (0.91) & $-22.94$ (0.85) \\
 & minimum & $+9.99$ (0.28) & $+18.85$ (0.42) & $+19.12$ (0.41) \\
\bottomrule
\end{tabular}
\end{table}

\paragraph{Refinements and sampling.}
\Cref{tab:refinements} refines the SPX grid and the number of cycles at $400{,}000$ paths, and the latent and emission grids with $10{,}000{,}000$ paths, where the standard errors are about $0.17$ basis points ($0.6$ for the smooth claim).
\begin{itemize}
    \item \textbf{Maximum and minimum claims:} over every Gauss--Hermite pair with $n_Z\ge7$ and every equal-probability pair with $n_Z\ge12$, at either path count, the differences lie in $-25.26$ to $-20.06$ and $+71.27$ to $+76.02$ basis points, and their memory terms in $-39.35$ to $-33.95$ and $+44.21$ to $+46.71$.
    \item \textbf{Smooth and average claims:} their memory terms are negative on every grid, $-79.15$ to $-43.34$ and $-26.53$ to $-15.46$ basis points over the same pairs, but their size moves with $n_Z$ and $n_W$; the smooth claim's difference ranges from $-7.60$ to $+20.15$.
    \item \textbf{Fit:} on every variant but two, the largest errors at the calibration strikes are within $0.03$ volatility points of those of the baseline pair.
    With $150$ cycles the band is fitted more closely and the wings less: for $Q_1$ the largest errors at the calibration strikes fall from $0.31$ to $0.20$ (SPX) and from $0.18$ to $0.12$ (VIX), the first window's log-contract miss falls from $1.48$ to $0.07$, and the mean SPX error over all quotes rises from $0.93$ to $1.11$ at the worst maturity.
    The maximum and minimum claims' differences move by less than $0.5$ basis points.
    With $n_W=24$ the largest VIX error at a calibration strike is $0.43$ volatility points for $Q_1$ and $0.51$ for $Q_0$, at $28$ days, against $0.18$ and $0.15$ for the baseline pair.
    The mean VIX error over all quotes, at the worst window, is $0.39$ and $0.39$ with $n_W=16$ and $0.35$ and $0.35$ with $n_W=24$, against $0.78$ and $0.79$ for the baseline pair.
    With $n_W=16$ the lowest VIX atom at $28$ days lies only $0.03\%$ of the forward below the strike at $0.70$ times the forward, so no law on these atoms can match the put's price there; that strike is left out, as on the baseline grid, and the first date keeps four strikes.
\end{itemize}

\begin{table}[ht]
\centering
\caption{The non-overlapping calendar under refinements: $P_{Q_1}-P_{Q_0}$ (net) and the memory term $P_{Q_1}-P_{\mathsf MQ_1}$, in basis points of notional.
Rows are labeled by latent grid and $(n_Z,n_W)$; equal-probability (EQ) or Gauss--Hermite (GH) latent nodes.}
\label{tab:refinements}
\footnotesize
\begin{tabular}{@{}lrrrrrrrr@{}}
\toprule
 & \multicolumn{2}{c}{maximum} & \multicolumn{2}{c}{minimum} & \multicolumn{2}{c}{smooth} & \multicolumn{2}{c}{average} \\
\cmidrule(lr){2-3}\cmidrule(lr){4-5}\cmidrule(lr){6-7}\cmidrule(lr){8-9}
 & net & memory & net & memory & net & memory & net & memory \\
\midrule
\multicolumn{9}{l}{\emph{400,000 paths}} \\
GH $(7, 12)$, the baseline law & $-22.12$ & $-36.25$ & $+72.16$ & $+45.49$ & $+11.33$ & $-58.10$ & $+4.54$ & $-20.95$ \\
GH $(7, 12)$, $n_S = 80$ & $-21.95$ & $-36.10$ & $+71.27$ & $+45.47$ & $+11.06$ & $-57.19$ & $+4.35$ & $-20.83$ \\
GH $(7, 12)$, 150 cycles & $-22.03$ & $-36.20$ & $+72.50$ & $+45.82$ & $+11.61$ & $-58.03$ & $+4.79$ & $-20.96$ \\
\addlinespace
\multicolumn{9}{l}{\emph{10,000,000 paths}} \\
EQ $(3, 12)$ & $-8.41$ & $-16.96$ & $+58.98$ & $+38.80$ & $+52.39$ & $-3.18$ & $+11.59$ & $-5.64$ \\
EQ $(9, 12)$ & $-19.92$ & $-33.14$ & $+71.37$ & $+44.66$ & $+26.79$ & $-48.81$ & $+7.00$ & $-18.28$ \\
EQ $(12, 12)$ & $-21.35$ & $-34.92$ & $+72.98$ & $+45.44$ & $+20.15$ & $-56.87$ & $+5.86$ & $-20.25$ \\
EQ $(15, 12)$ & $-22.20$ & $-36.01$ & $+73.74$ & $+45.84$ & $+15.89$ & $-62.47$ & $+5.06$ & $-21.67$ \\
GH $(3, 12)$ & $-39.89$ & $-56.43$ & $+56.38$ & $+32.01$ & $-37.63$ & $-117.16$ & $-5.30$ & $-34.92$ \\
GH $(7, 12)$ & $-22.33$ & $-36.03$ & $+71.66$ & $+44.34$ & $+15.62$ & $-60.69$ & $+5.44$ & $-21.15$ \\
GH $(9, 12)$ & $-23.18$ & $-37.27$ & $+74.57$ & $+46.17$ & $+5.80$ & $-72.52$ & $+3.47$ & $-24.05$ \\
GH $(11, 12)$ & $-23.94$ & $-38.43$ & $+75.59$ & $+46.71$ & $+0.70$ & $-78.88$ & $+2.31$ & $-25.72$ \\
GH $(7, 16)$ & $-20.06$ & $-33.95$ & $+73.26$ & $+44.21$ & $+7.93$ & $-57.31$ & $+6.41$ & $-18.00$ \\
GH $(9, 16)$ & $-21.46$ & $-35.85$ & $+76.02$ & $+45.76$ & $-4.82$ & $-71.43$ & $+4.04$ & $-21.28$ \\
GH $(7, 24)$ & $-23.97$ & $-36.94$ & $+74.85$ & $+44.82$ & $+18.80$ & $-43.34$ & $+9.56$ & $-15.46$ \\
\bottomrule
\end{tabular}
\end{table}

\paragraph{All variants.}
\Cref{fig:variants} collects the maximum and minimum claims' differences of \cref{tab:convergence,tab:refinements,tab:crosscal,tab:band}, each law once.
Every one has the same sign as for the baseline pair and exceeds both $5$ basis points and four paired standard errors: the maximum claim's difference ranges from $-40.28$ to $-6.75$ basis points and the minimum claim's from $+9.99$ to $+99.28$.

\begin{figure}[tbp]
\centering
\includegraphics[width=\textwidth]{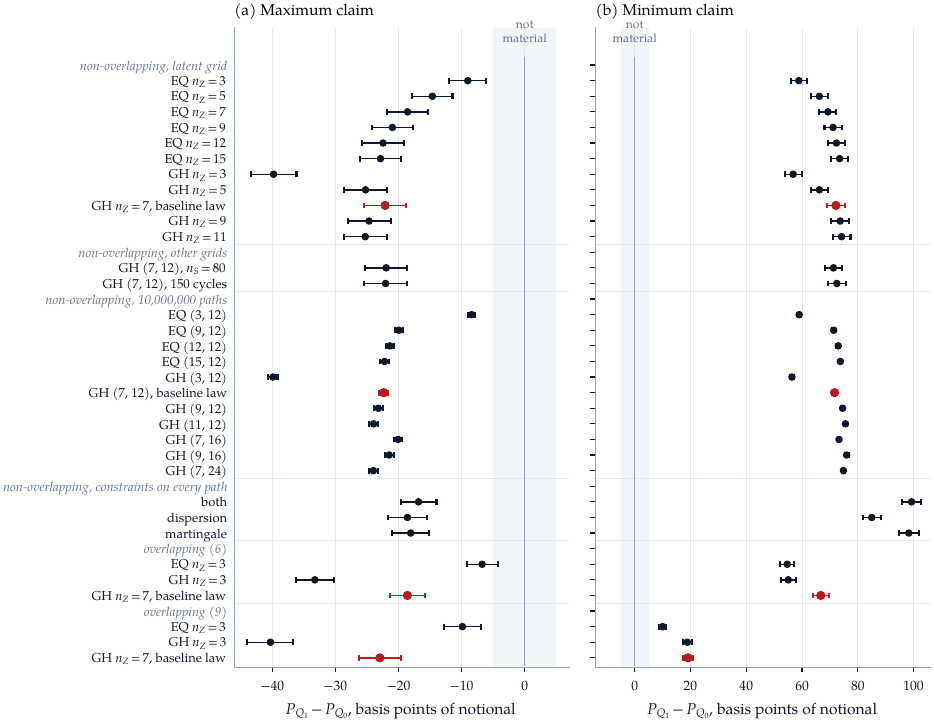}
\caption{$P_{Q_1}-P_{Q_0}$ for the maximum and minimum claims under every law of \cref{tab:convergence,tab:refinements,tab:crosscal,tab:band}, each once, with bars of four paired standard errors; $400{,}000$ paths unless stated.
EQ and GH are the equal-probability and Gauss--Hermite latent grids, and $(n_Z,n_W)$ the latent and emission grid sizes.
The baseline laws are in red; the grey band marks differences below $5$ basis points, which are not material.}
\label{fig:variants}
\end{figure}

\paragraph{Fit by maturity on the overlapping calendars.}
\Cref{fig:fit_m6m,fig:fit_m9m} repeat \cref{fig:fit} for the six- and nine-maturity overlapping calendars.
On both calendars the largest error of $Q_1$ at a calibration strike is at $58$ days on the SPX and at $57$ days on the VIX.
With six maturities it is $0.66$ (SPX) and $1.42$ (VIX), against $0.63$ and $1.18$ for $Q_0$; with nine it is $0.67$ and $1.18$, against $0.65$ and $1.04$.

\begin{figure}[t]
\centering
\includegraphics[width=\textwidth]{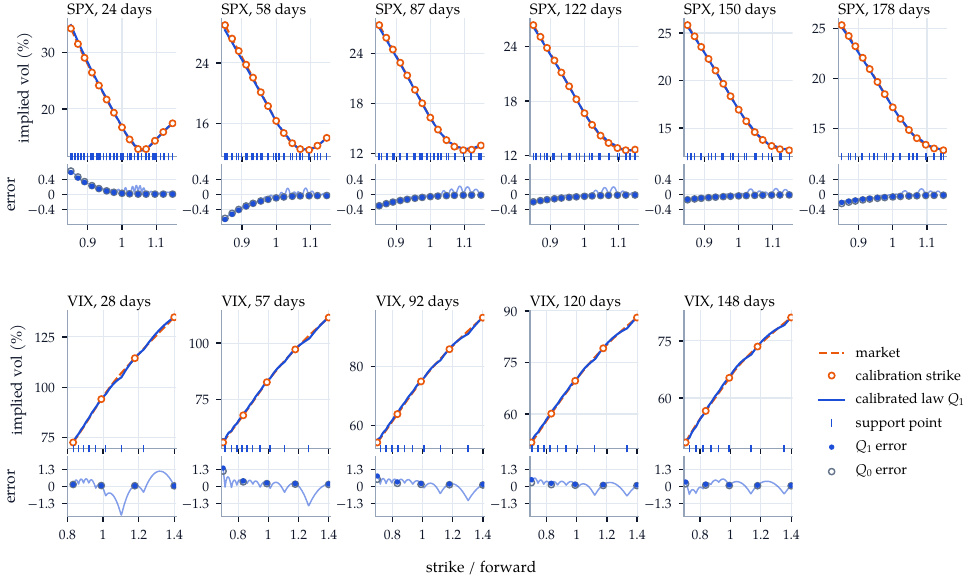}
\caption{Fit of the laws on the six-maturity overlapping calendar at every SPX maturity (top) and VIX date (bottom), over the range of calibration strikes only, drawn as in \cref{fig:fit}.}
\label{fig:fit_m6m}
\end{figure}

\begin{figure}[p]
\centering
\includegraphics[width=\textwidth]{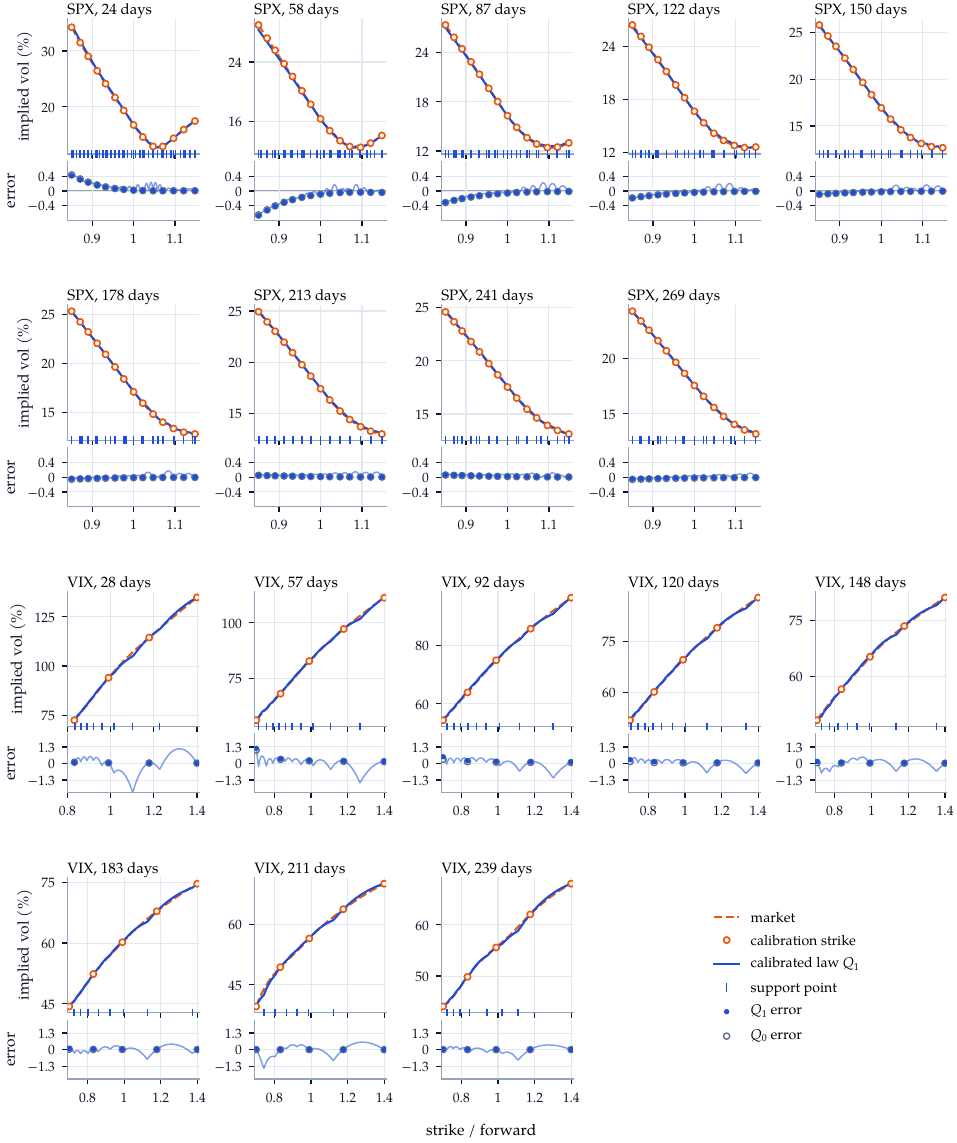}
\caption{Fit of the laws on the nine-maturity calendar at every SPX maturity (top two rows) and VIX date (bottom two rows), over the range of calibration strikes only, drawn as in \cref{fig:fit}.}
\label{fig:fit_m9m}
\end{figure}

\paragraph{Joint law of consecutive VIX levels.}
\Cref{fig:vixpairs} shows the joint law of the VIX at consecutive dates on the non-overlapping calendar, relative to independence (\cref{subsec:market_prices}).

\begin{figure}[t]
\centering
\includegraphics[width=\textwidth]{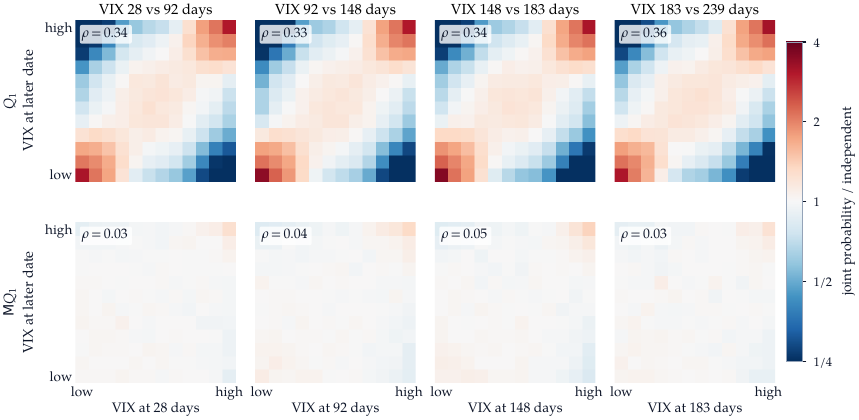}
\caption{Joint law of the VIX at consecutive dates on the non-overlapping calendar, over the $12$ atoms of each date from low to high, divided by the product of its two marginals ($1$ is independence; log scale), under $Q_1$ (top) and $\mathsf MQ_1$ (bottom); $\rho$ is the correlation of the VIX levels.}
\label{fig:vixpairs}
\end{figure}

\paragraph{All claims on all calendars.}
\Cref{tab:allclaims} gives the four claims on the three calendars at $n_Z=7$.
The maximum and minimum claims carry the comparison: memory moves them in one direction, and their differences keep their signs on every grid and calendar (\cref{fig:variants}).
The smooth and average claims depend on the VIX average, where the memory term and the change of the blocks are large and of opposite sign (\cref{tab:memory}); their net differences are therefore small sums of large terms, material on the overlapping calendars and dependent on the latent grid on the non-overlapping one (\cref{tab:convergence,tab:refinements}), and we report them for completeness.

\begin{table}[ht]
\centering
\caption{$P_{Q_1}-P_{Q_0}$ for the four claims on each calendar at $n_Z=7$, in basis points of notional, with paired standard errors ($400{,}000$ paths).}
\label{tab:allclaims}
\small
\begin{tabular}{@{}lrrr@{}}
\toprule
claim & non-overlapping & overlapping (6) & overlapping (9) \\
\midrule
maximum & $-22.12$ (0.84) & $-18.58$ (0.70) & $-22.94$ (0.85) \\
minimum & $+72.16$ (0.78) & $+66.73$ (0.70) & $+19.12$ (0.41) \\
smooth & $+11.33$ (2.87) & $+31.96$ (2.63) & $+115.34$ (2.85) \\
average & $+4.54$ (0.86) & $+7.41$ (0.73) & $+28.85$ (0.89) \\
\bottomrule
\end{tabular}
\end{table}

\paragraph{The band.}
\Cref{tab:gap} compares two volatilities of the input surfaces at each VIX date $T$.
The first is the root-mean-square VIX, $\sqrt{\Embb[V^2]}$, with $\Embb[V^2]=F_V^2+2\int\mathrm{OTM}(K)\,dK$ over the quoted VIX strikes.
The second is the $30$-day forward volatility of the SPX over its quoted strikes, $\bigl((w(T+\tau)-w(T))/\tau\bigr)^{1/2}$ with $\tau=30/365$ and the log-contract value $w=2\int\mathrm{OTM}(K)\,K^{-2}\,dK$; the snapshot quotes SPX options only at nine monthly expiries, so $w$ at a VIX date is interpolated linearly in time between the two around it.
\begin{itemize}
    \item \textbf{The gap} is $1.65$ to $6.16$ volatility points and widens with the date, and the ranges from bid and ask prices never overlap.
    By~\eqref{eq:consistency_global} and this interpolation, $14.6$ to $43.3\%$ of $\Embb[V^2]$ then comes from SPX strikes below $0.67$ or above $1.31$ to $1.49$ times the forward, where there are no quotes.
    \item \textbf{Paths that leave the band} before the last step are $27.4$ to $48.7\%$ of the paths under $Q_1$; they carry $85.1$ to $98.5\%$ of the maximum and minimum claims' prices, and $52$ to $99\%$ of their differences $P_{Q_1}-P_{Q_0}$.
    These claims pay only when the SPX average ends low.
    \item \textbf{Conditional constraints on every path:} \cref{tab:band} recalibrates the pair of the non-overlapping calendar with both conditional families on every path, or with one of them on every path and the other until the SPX first leaves the band; every variant meets its conditional constraints, with $\mathfrak R_{\mathrm{cond}}<10^{-4}$.
    Each fits the calibration strikes worse: for $Q_1$ the largest errors are $1.06$ to $1.18$ volatility points on the SPX and $0.66$ to $1.39$ on the VIX, against $0.31$ and $0.18$.
    Each keeps the signs and materiality of the maximum and minimum claims' differences and the signs of their memory terms.
    With the dispersion constraints only in the band, the log-contract constraint of the $92$-day window misses by $2.15$ volatility points.
    The smooth and average claims' differences rise to $+21.87$ to $+30.44$ and $+5.22$ to $+7.00$ basis points; as on the grids of \cref{tab:refinements}, their net differences move with the setup, and the comparison rests on the maximum and minimum claims.
\end{itemize}

\begin{table}[ht]
\centering
\caption{The input surfaces at each VIX date: the root-mean-square VIX and the $30$-day forward volatility of the SPX over its quoted strikes, in volatility points, at mid prices with the range from bid and ask prices; their gap; and the share of $\Embb[V^2]$ that the quoted SPX strikes do not supply, in percent.}
\label{tab:gap}
\small
\begin{tabular}{@{}rrrrr@{}}
\toprule
VIX date (days) & $\sqrt{\Embb[V^2]}$ & SPX forward volatility & gap & beyond the quoted strikes \\
\midrule
28 & 21.84 (21.82--21.86) & 20.19 (19.78--20.58) & 1.65 & 14.6 \\
57 & 22.76 (22.73--22.79) & 20.25 (19.70--20.79) & 2.51 & 20.8 \\
92 & 23.71 (23.68--23.74) & 20.78 (20.02--21.52) & 2.93 & 23.2 \\
120 & 23.98 (23.95--24.02) & 20.76 (19.87--21.62) & 3.22 & 25.1 \\
148 & 24.38 (24.34--24.42) & 20.04 (19.00--21.03) & 4.34 & 32.4 \\
183 & 24.82 (24.77--24.87) & 20.48 (19.36--21.53) & 4.34 & 31.9 \\
211 & 24.84 (24.66--25.02) & 19.67 (18.44--20.83) & 5.18 & 37.3 \\
239 & 24.91 (24.66--25.16) & 18.75 (17.32--20.08) & 6.16 & 43.3 \\
\bottomrule
\end{tabular}
\end{table}

\begin{table}[ht]
\centering
\caption{The non-overlapping calendar with the conditional constraints of neither family (the baseline pair), both families or one of them on every path, the rest only until the SPX first leaves the band.
$P_{Q_1}-P_{Q_0}$ (net) with paired standard errors and the memory term $P_{Q_1}-P_{\mathsf MQ_1}$, in basis points of notional; for $Q_1$, the largest implied-volatility errors at the calibration strikes and the largest miss of a window's log-contract constraint, in volatility points.}
\label{tab:band}
\footnotesize
\begin{tabular}{@{}lrrrrrrr@{}}
\toprule
 & \multicolumn{2}{c}{maximum} & \multicolumn{2}{c}{minimum} & \multicolumn{2}{c}{calibration, max} & \\
\cmidrule(lr){2-3}\cmidrule(lr){4-5}\cmidrule(lr){6-7}
constraints on every path & net & memory & net & memory & SPX & VIX & log-contract \\
\midrule
neither (the baseline pair) & $-22.12$ (0.84) & $-36.25$ & $+72.16$ (0.78) & $+45.49$ & 0.31 & 0.18 & 1.48 \\
both & $-16.82$ (0.71) & $-35.65$ & $+99.28$ (0.85) & $+55.34$ & 1.18 & 0.66 & 0.00 \\
dispersion & $-18.58$ (0.77) & $-32.80$ & $+85.01$ (0.81) & $+48.81$ & 1.09 & 1.39 & 0.00 \\
martingale & $-18.07$ (0.74) & $-33.16$ & $+98.30$ (0.87) & $+50.21$ & 1.06 & 0.93 & 2.15 \\
\bottomrule
\end{tabular}
\end{table}

\clearpage
\section{Computation}\label[appendix]{app:computation}

\Cref{tab:sizes} gives the size of each calibration: the imposed conditional constraints and the coefficients of the law, against the entries of the joint law of the SPX and VIX on the same grid, which grow exponentially with the number of maturities: about $10^{25}$ with six SPX maturities and $10^{39}$ with nine.
The step kernels, which dominate the work, are evaluated without forming a block's four-dimensional array.
Each law of \cref{sec:experiments} was calibrated in under an hour on one multi-core server.
A baseline law took $26$ to $27$ minutes with $4$ cores and $0.5$~GB of memory on the non-overlapping calendar, $27$ minutes with $64$ cores and $80$~GB on the six-maturity overlapping calendar, and $42$ minutes with $64$ cores and $125$~GB on the nine-maturity one.
The finer grids of \cref{app:discretization} take longer; these times describe our implementation and are not a benchmark.

\begin{table}[ht]
\centering
\caption{Size of each calibration on Gauss--Hermite latent grids, the same for both memory orders: imposed conditional constraints, coefficients of the law, and entries of the joint law of the SPX and VIX on the same grid.}
\label{tab:sizes}
\small
\begin{tabular}{@{}lcrrr@{}}
\toprule
calendar & $n_Z$ & conditional constraints & coefficients & joint-law entries \\
\midrule
non-overlapping & 3 & 5,340 & 42,544 & $1.6\times10^{25}$ \\
overlapping (6) & 3 & 998,211 & 50,493,700 & $1.6\times10^{25}$ \\
overlapping (9) & 3 & 1,135,941 & 84,153,322 & $1.7\times10^{39}$ \\
\addlinespace
non-overlapping & 7 & 12,460 & 99,104 & $1.6\times10^{25}$ \\
overlapping (6) & 7 & 5,296,711 & 269,550,692 & $1.6\times10^{25}$ \\
overlapping (9) & 7 & 6,013,217 & 449,244,538 & $1.7\times10^{39}$ \\
\bottomrule
\end{tabular}
\end{table}

\paragraph{Penalty route.}
When a conditional target cannot be attained on the grid, the exact projection of \cref{subsec:solver} does not exist.
Write the calibration as the minimization of $D_{\mathrm{KL}}(\mu\Vert\bar\mu)$ over laws $\mu$ on the grid subject to the market constraints $A_h\mu=b_h$ and the conditional constraints $A_c\mu=b_c$, each conditional constraint scaled as in~\eqref{eq:dw_residual}.
An augmented-Bregman scheme keeps the projections onto the market constraints and penalizes the conditional ones: its conditional step is the entropic mirror step~\citep{nemirovski1983problem,beck2003mirror}
\[
\mu\leftarrow\frac{\mu\odot e^{-\eta g}}{\langle\mu,e^{-\eta g}\rangle},
\qquad
g=A_c^{\top}\bigl[y+\lambda(A_c\mu-b_c)\bigr],
\]
with step size $\eta$, penalty weight $\lambda$ and multipliers updated as $y\leftarrow y+\lambda(A_c\mu-b_c)$.
\Cref{prop:penalty_limit} gives the limit of the penalized problem as the weight grows; it concerns the solutions of that problem, not the iterates of the scheme after a fixed number of steps.

\begin{proposition}[Penalty limit on a finite grid]\label{prop:penalty_limit}
Let $\mathcal H:=\{\mu:A_h\mu=b_h\}$ be a nonempty set of laws on the finite grid, $\bar\mu$ positive on the grid, $r(\mu):=\tfrac12\lVert A_c\mu-b_c\rVert^2$ and, for $\lambda>0$,
\[
\mu_\lambda:=\argmin_{\mu\in\mathcal H}\bigl\{D_{\mathrm{KL}}(\mu\Vert\bar\mu)+\lambda r(\mu)\bigr\}.
\]
Then $\mu_\lambda$ is unique.
As $\lambda\to\infty$, $r(\mu_\lambda)\to r_\star:=\min_{\mathcal H}r$, and $\mu_\lambda$ converges to the element of $\argmin_{\mathcal H}r$ with the least relative entropy.
If every constraint can be met, $r_\star=0$ and the limit is the exact projection; otherwise it is the law that misses the conditional constraints least while keeping the market constraints.
\end{proposition}

\begin{proof}
The set $\mathcal H$ is a closed subset of the finite-dimensional simplex, hence compact, and the objective is continuous on it and strictly convex because $\bar\mu$ is positive; so $\mu_\lambda$ exists and is unique.
For any $\mu_\star\in\argmin_{\mathcal H}r$, optimality of $\mu_\lambda$ gives
\[
D_{\mathrm{KL}}(\mu_\lambda\Vert\bar\mu)+\lambda r(\mu_\lambda)\le D_{\mathrm{KL}}(\mu_\star\Vert\bar\mu)+\lambda r_\star.
\]
Relative entropy is nonnegative, and finite on the grid because $\bar\mu$ is positive, so $0\le r(\mu_\lambda)-r_\star\le D_{\mathrm{KL}}(\mu_\star\Vert\bar\mu)/\lambda$ and $r(\mu_\lambda)\to r_\star$.
The same inequality and $r(\mu_\lambda)\ge r_\star$ give $D_{\mathrm{KL}}(\mu_\lambda\Vert\bar\mu)\le D_{\mathrm{KL}}(\mu_\star\Vert\bar\mu)$ for every such $\mu_\star$.
By compactness $(\mu_\lambda)$ has limit points; by continuity each lies in $\argmin_{\mathcal H}r$ and has relative entropy at most that of every element of this set.
The set is convex and relative entropy is strictly convex on it, so its element of least relative entropy is unique, every limit point equals it, and the whole family converges to it.
The two cases follow from whether $r_\star$ is zero.
\end{proof}

\section*{Acknowledgments}

We thank Junhyung Lyle Kim and Rudy Raymond for discussions and feedback.
We thank Ruslan Shaydulin, Rob Otter and our colleagues in Global Technology Applied Research and Quantitative Trading \& Research at JPMorganChase for their support.

\section*{Disclaimer}
This paper was prepared for informational purposes by the Global Technology Applied Research center and the Quantitative Trading \& Research group of JPMorgan Chase \& Co. This paper is not a product of the Research Department of JPMorgan Chase \& Co. or its affiliates. Neither JPMorgan Chase \& Co. nor any of its affiliates makes any explicit or implied representation or warranty and none of them accept any liability in connection with this paper, including, without limitation, with respect to the completeness, accuracy, or reliability of the information contained herein and the potential legal, compliance, tax, or accounting effects thereof. This document is not intended as investment research or investment advice, or as a recommendation, offer, or solicitation for the purchase or sale of any security, financial instrument, financial product or service, or to be used in any way for evaluating the merits of participating in any transaction.

\bibliographystyle{plainnat}
\bibliography{refs}

\end{document}